\documentclass[11pt]{article}
\usepackage{amsmath,amssymb,amsthm}
\usepackage{graphicx}
\usepackage{booktabs}
\usepackage{placeins}
\usepackage{hyperref}
\usepackage{geometry}
\usepackage{microtype}
\usepackage{comment}
\newtheorem{proposition}{Proposition}

\theoremstyle{definition}

\newtheorem{example}{Example}
\theoremstyle{remark}
\newtheorem{remark}{Remark}
\newtheorem{conjecture}{Conjecture}

\newcommand{\dfn}{\stackrel{\triangle}{=}}
\newcommand{\reals}{\mathbb{R}}
\newcommand{\bE}{\mathbf{E}}
\newcommand{\bx}{\mathbf{x}}
\newcommand{\tbx}{\tilde{\mathbf{x}}}
\newcommand{\tbX}{\tilde{\mathbf{X}}}
\newcommand{\bX}{\mathbf{X}}

\newcommand{\calM}{\mathcal{M}}
\newcommand{\calQ}{\mathcal{Q}}
\newcommand{\calS}{\mathcal{S}}
\newcommand{\calX}{\mathcal{X}}
\newcommand{\Var}{\mathrm{Var}}
\newcommand{\Cov}{\mathrm{Cov}}
\newcommand{\MSE}{\mathrm{MSE}}

\title{\textbf{Grouped Reverse Importance Sampling for the Partition Function}}
\author{Neri Merhav\\[6pt]
\normalsize The Viterbi Faculty of Electrical and Computer Engineering\\
\normalsize Technion -- Israel Institute of Technology\\
\normalsize Technion City, Haifa 3200003, Israel\\
\normalsize \texttt{merhav@technion.ac.il}}
\date{}

\begin{document}
\maketitle
\thispagestyle{empty}

\begin{abstract}

We introduce and analyze several grouped variants of the method of reverse 
importance sampling (RIS) for estimating a partition function $Z(\beta)=\int e^{-\beta U(x)}\,\mbox{d}x$
from samples of the Boltzmann distribution $p(x)=e^{-\beta U(x)}/Z(\beta)$.
Ordinary RIS weighs each sample separately. By contrast, our proposed grouped RIS
(GRIS) methods are based on assigning the samples into groups (or
batches) of size $k\ge 2$ and applying a joint weight function to each group.
The focal point of the research is the quest for a tractable weight function that would
yield the smallest possible mean squared error (MSE). A simple identity 
relates the normalized MSE to the chi-squared divergence between the 
joint-weight distribution and the distribution of the $k$-fold sum 
of independent energies. 
Our first theoretical finding is that any weight that improves on ordinary RIS
($k=1$) must couple the group components. In other words, it must not be a
product-form function across those components,
as product-form weight functions always worsen the MSE.
Our second, and more important, finding
is that, without loss of optimality, it is sufficient
to seek weight functions that depend only on the total energy,
$\sum_iU(x_i)$, of the group (group-energy weight functions);
for the sliding-window variants, the analogous result is open.
This finding simplifies both the theoretical analysis and the application
of the method substantially.
For $k=2$ and $k=3$, the MSE associated with non-overlapping (NOL) groups is 
reduced by $20$--$65\%$ across three examples. 
We then propose two additional variants of GRIS, both
based on sliding-window grouping (as opposed to NOL
grouping). The first applies a fixed weight sliding window (FSW) across all (cyclic) shifts
of the sliding window, and the second allows a variable-weight sliding window
(VSW). The FSW scheme improves
on the NOL one, and the VSW improves 
even further, as will be demonstrated numerically.\\

\end{abstract}

\section{Introduction}
\label{sec:intro}

\subsection{A Few Words of Background}

Computing the partition function (or the normalizing constant) of a Boltzmann
distribution is a fundamental task in statistical physics, Bayesian
statistics, and machine learning.  It arises whenever one needs to evaluate
the probability of evidence, compare models, or compute thermodynamic
quantities such as free energy, average energy, pressure, etc.  The partition function $Z(\beta)$ is the
integral of the Boltzmann factor $e^{-\beta U(x)}$ over the
state space, where $U(x)$ is the Hamiltonian (energy function) 
associated with a microstate $x\in\calX$ and the parameter $\beta$ is called the inverse
temperature.  In virtually all non-trivial applications, this integral is
analytically intractable with no apparent closed form representation, especially
when the state space $\calX$ is high-dimensional and
the energy landscape is complicated. Consequently, one frequently resorts to Monte Carlo methods.

Reverse importance sampling (RIS)
is a natural approach when one already has access to samples from the
Boltzmann distribution, $p(x)=e^{-\beta U(x)}/Z(\beta)$.
The idea is to draw $n$ independent samples, $x_1,\ldots,x_n$, from $p(\cdot)$
and reweight each one by a function $m(x)$ chosen so that its own
normalizing constant, $M = \int_{\calX} m(x)dx$, is computable
in closed form, or at least easy to calculate numerically by integration over
a space with low dimension. The resulting estimate of
$\ln Z(\beta)$ turns out to be $\ln M$ minus the logarithm of the sample mean
of the importance weights $m(x_i)e^{\beta U(x_i)}$, $i=1,2,\ldots,n$.
The estimator is consistent for any positive weight function $m$,
and its mean squared error (MSE) depends on how well $m(x)$
approximates the ideal weight function $m^*(x)=e^{-\beta U(x)}$,
which gives zero MSE, but is impractical since it
would require computing $M=\int_{\calX} m^*(x)\,\mbox{d}x=Z(\beta)$,
which we presumably cannot calculate in the first place.
The design problem is therefore to choose $m$ to minimize the MSE across a certain family 
of functions that lend themselves to reasonably easy evaluation of $M$ -- henceforth 
referred to as tractable functions. The MSE formula is well-known to be
representable in terms of the chi-squared divergence between two disributions, one of
which is associated with the choice of $m$ and the other is derived from the
Boltzmann disribution $p$.

\subsection{Main Contributions}

In this paper, we focus on the following question: \emph{can one do better by
grouping the $n$ samples into batches of size $k$ and applying a
single joint weight to each batch?}
We call this \emph{grouped RIS} (GRIS). The idea is that the $k$ samples in each
group are processed jointly, in other words, the GRIS estimator
assigns a single weight $m(\bx)$, $\bx\in\calX^k$, to each group and estimates
$\ln Z^k(\beta)=k\ln Z(\beta)$. Thus, division by $k$ gives an estimate
of $\ln Z(\beta)$.
At first glance, one may wonder: how can grouping possibly help when 
the samples are statistically independent?
The key idea is that a joint weight function allows a significantly richer freedom
for optimization that cannot be exploited by weight functions that operate on individual samples.
From the perspective of information theorists, this is 
similar to the fact that source- and channel coding in long blocks
outperforms scalar coding even when the system is memoryless.

Stated briefly, we develop three successive tiers of improvement, each building
on the previous. First, we focus on \emph{non-overlapping} (NOL) grouping: partition
the given set of
$n$ samples into $n/k$ disjoint groups, each of size $k$, and apply a joint
weight to each one. For $k=2$ and $k=3$, the MSE associated with NOL groups is
reduced by $20$--$65\%$ relative to ordinary RIS ($k=1$), across three examples.
Next, we consider a \emph{sliding-window} variant: form $n$ overlapping groups
by cyclically sliding a window of size $k$, one step at a time, using the same
weight function throughout. We refer to this variant of GRIS as {\em fixed-weight
sliding-window} (FSW) grouping.
We furnish a sufficient condition
on the overlap correlations under which this FSW scheme reduces the MSE further,
and verify the condition across the three examples with different energy
functions (Sections~\ref{sec:overlap}--\ref{sec:numerics}).
In all three examples the condition is satisfied and the sliding window
achieves an additional $13$--$14\%$ reduction in MSE over non-overlapping
grouping. Third, we propose a \emph{variable-weight sliding window} (VSW)
variant of GRIS:
cycle through $k$ different weight functions --- one per group position in
every cycle. This yields a further gain of $30$--$40\%$ over the FSW,
as demonstrated in the numerical examples.

In order to obtain the above improvements, we first have to
build the theoretical basis of this work. To this end, we establish the
following general results that set the stage for providing guidelines
regarding the choice of good weight functions for NOL GRIS.\\

\noindent
1.~{\em Coupling is necessary for NOL GRIS.}
A necessary condition for a joint weight to improve on the $k=1$
baseline (standard RIS) is that it must genuinely create coupling among the
group components: the weight
must depend on a group, say, $\bx=(x_1,\ldots,x_k)$ in a way that does not factor as a
product of the individual weights. Concretely,
we show that any product-form weight
function actually \emph{worsens} the MSE relative to $k=1$, so coupling is
not merely helpful but essential. This motivates the search for
joint weights that create dependence between the group components.\\

\noindent
2.~{\em Group-energy weight functions are sufficient without loss of optimality
for NOL GRIS.}
We prove (Section~\ref{sec:ge_optimal}) that for the NOL estimator,
for any positive weight function $m(x_1,\ldots,x_k)$, there exists a
group-energy weight function $\tilde m(U(x_1)+\cdots+U(x_k))$ with MSE
no larger than that of $m(x_1,\ldots,x_k)$: the group energy plays the
role of a sufficient statistic.
We use group-energy weights for the FSW and VSW variants as well,
for tractability; the analogous sufficiency result in those settings
is an open problem (Section~\ref{sec:conclusions}).\\

\noindent
3.~{\em Dimension reduction.} An additional benefit of item 2 above, 
is that it enables one to reduce the calculation of the normalizing constant $M$ from a $k$-dimensional 
integration to a one-dimensional integral
independently of $k$. This means that even if
$M$ has no apparent closed-form
expression, it can often still be evaluated at least numerically at a computational effort that
does not depend on $k$, most commonly -- by integration in one
dimension. To this end, we need to evaluate the measure of the set of
vectors $(x_1,\ldots,x_k)$ whose group energy, $U(x_1)+\cdots+U(x_k)$,
equals a given value $u$, a purely geometric quantity depending only
on the energy function $U$, which we call the
\emph{density of states} $\Omega_k(u)$, a term borrowed from the
realm of statistical physics. Computing $\Omega_k(u)$ is possible
in many cases of interest, as discussed in Section~\ref{sec:ge_optimal}.
Moreover, upon passing to integration in one
dimension, it may suffice to parametrize the weight function by a relatively small number of
parameters to be optimized.\\

\noindent
4.~{\em When is GRIS genuinely useful?}
An observant reader may notice a potential circularity: if the
density of states $\Omega_k(u)$ is tractable, one can in
principle compute $Z(\beta)$ directly from it via the one-dimensional
integral, $Z(\beta)=\int_0^\infty e^{-\beta u}\Omega_k(u)\mbox{d}u$ (a Laplace
transform), and then RIS and GRIS are not needed in the first place. The genuine use case for GRIS
is therefore the \emph{perturbed Hamiltonian} setting, where
$U(x)=U^\star(x)+\epsilon V(x)$ ($\epsilon$ being a small parameter), and where
the density of states of the main term
$U^\star(\cdot)$ is tractable and serves as the basis for the weight function,
but the density of states of the full energy function $U(\cdot)$ is not --- for
instance when $V(x)$ introduces complicated cross-terms that couple coordinates.
In this setting, $Z(\beta)$ is genuinely intractable, yet GRIS with a
main-term weight provides a consistent estimate at a controlled MSE
cost. To keep the theoretical analysis transparent, the first three
numerical examples in Section~\ref{sec:numerics} are \emph{toy examples}
--- one-dimensional cases where $Z(\beta)$ is in fact tractable ---
chosen solely as controlled experiments that allow exact MSE analysis for
verification without Monte Carlo noise. They can also be viewed as
describing the limiting behavior of the perturbed Hamiltonian setting
as $\epsilon\to 0$: when the perturbation vanishes, $Z(\beta)$ becomes
tractable and the toy examples are recovered exactly. The method,
however, works for every value of $\epsilon$, not merely very small
$\epsilon>0$. The genuinely intractable
setting is discussed in Section~\ref{sec:perturbed}.\\

We also study grouped forward importance sampling (FIS), where samples are drawn from a
proposal distribution rather than from $p(\cdot)$ (Appendix~A).
The same chi-squared structure appears, but the sliding window fails:
in RIS the weight function and the sampling distribution are independent
(samples come from $p(\cdot)$ regardless of $m(\cdot)$), so overlapping groups
are unbiased; in FIS the proposal generates the samples and
hence the weight depends on the same draw, so overlapping groups
with a non-product proposal are necessarily biased.
This structural asymmetry explains why GRIS is the
more natural and complete theory.\\

\subsection{Related Work}

A few words on relevant earlier work are in order.\\

\noindent
1.~{\em RIS and related estimators.}
Drawing samples from the target and reweighting to estimate a normalizing
constant according to the RIS paradigm is intimately related to the
harmonic mean estimator of Newton and Raftery~\cite{NewtonRaftery1994}, which is
the special case where the weight is the reciprocal of the unnormalized
density; it is notoriously unstable because the resulting importance weights
may have infinite variance. Bridge sampling~\cite{MengWong1996} uses samples from
two distributions simultaneously to estimate a ratio of normalizing constants.
Discriminance sampling~\cite{LiuFang2015} converts partition function
estimation into a classification problem; this is also the paper where
the $k=1$ chi-squared interpretation of the RIS variance is first
stated explicitly.  Both methods apply weights to individual samples rather than to
groups, and neither considers joint weight functions.\\

\noindent
2.~{\em Annealed importance sampling (AIS).}
AIS~\cite{Neal2001} constructs a sequence of intermediate distributions
bridging a tractable reference to the target, and runs a Markov chain
through them to accumulate an importance weight. Unlike GRIS,
AIS does not require samples from $p(\cdot)$: it is designed for settings where
direct sampling from $p(\cdot)$ is infeasible. The two methods are therefore
complementary rather than competing; we provide a detailed comparative review in
Appendix~B.\\

\noindent
3.~{\em Importance weighted autoencoders (IWAE).}
The IWAE~\cite{Burda2015} processes $k$ i.i.d.\ draws from a recognition
network jointly by taking the logarithm of their average importance weight.  The
gain over a single draw comes from Jensen's inequality applied to the
logarithm, not from a non-product joint weight: each individual weight is
still computed separately.  The IWAE gain and the GRIS gain are
therefore distinct mechanisms, though both exploit the nonlinearity of the
logarithmic function.\\

\noindent
4.~{\em Path sampling and thermodynamic integration.}
Gelman and Meng~\cite{GelmanMeng1998} give a unified treatment of IS,
bridge sampling, and path sampling (thermodynamic integration) for
normalizing constants, showing that all three are connected by a common
identity. Path sampling computes $\ln[Z(\beta)/Z(\beta_0)]$ as
$\int_{\beta_0}^\beta\bE_t[U(X)]\,\mbox{d}t$, which is exact but requires
evaluating the expected energy at every intermediate temperature --- a
different computational regime from GRIS, which needs only samples at a
single $\beta$.\\

\noindent
5.~{\em Sequential Monte Carlo (SMC).}
SMC samplers~\cite{DelMoralDoucetJasra2006} propagate a cloud of
weighted particles through a sequence of distributions, providing
consistent estimates of normalizing constants along the way.  Like AIS,
SMC does not require samples from $p(\cdot)$ and is designed for settings
where direct sampling from $p(\cdot)$ is infeasible; it can be seen as a
resampled, adaptive version of AIS. GRIS is complementary: it operates
in the regime where $p(\cdot)$ is directly samplable.\\

\noindent
6.~{\em Variance reduction for IS.}
Owen and Zhou~\cite{OwenZhou2000} study FIS for computing
expectations $\int f(x)p(x)\,\mbox{d}x$ and show that using a mixture of
proposal distributions as control variates bounds the asymptotic
variance by a small multiple of the best single-proposal variance.
Their setting (combining samples from several proposals to estimate
an expectation) is different from GRIS (grouping i.i.d.\ samples from a single Boltzmann
distribution to estimate a normalizing constant), but the shared theme
is systematic variance reduction through a richer function class.
The standard reference for IS, Markov Chain Monte Carlo (MCMC), and related Monte Carlo methods
is Robert and Casella~\cite{RobertCasella2004}.\\

\noindent
7.~{\em Coupled self-normalized IS.}
Branchini and Elvira~\cite{BranchiniElvira2024} reduce the variance of a
self-normalized IS ratio by coupling numerator and denominator samples via
a joint proposal over two draws.  Their setting is a ratio of two
expectations; ours is a single unnormalized integral, and our grouping
operates over $k$ draws contributing to a single estimate rather than
coupling numerator and denominator.\\

\noindent
8.~{\em Multiple importance sampling and $k$-tuple methods.}
Multiple IS~\cite{Veach1995,Elvira2019} combines samples from several
different proposals; its ``grouping'' refers to grouping proposals, not
draws from a single distribution. Some combinatorial methods for polymer partition functions
select $k$ distinct microstates jointly,
but their motivation is combinatorial and the draws are not i.i.d.\\

To the best of our knowledge, GRIS --- i.i.d.\ draws from a single Boltzmann
distribution, grouped with a joint weight, guided by some theoretical
principles, as described
above, has not been studied
previously.

\section{The GRIS Estimator with Non-Overlapping Grouping}
\label{sec:setup}

\subsection{Ordinary RIS}

To set the stage for background, we begin by introducing the ordinary form of
RIS. Let $x_1,\ldots,x_n$ be independent samples drawn from
the Boltzmann distribution
\begin{equation}\label{eq:boltzmann}
  p(x) = \frac{e^{-\beta U(x)}}{Z(\beta)}, \qquad
  Z(\beta) = \int_\calX e^{-\beta U(x)}\mbox{d}x,
\end{equation}
where $\calX$ designates the alphabet of $x$.
The goal is to estimate $\ln Z(\beta)$
in cases where it is not available in closed form.
The basic RIS estimator is based on selecting a positive \emph{weight function}
$m:\calX \to (0,\infty)$ whose normalizing constant $M=\int_\calX m(x)\mbox{d}x$ is available
in closed form, or at least by reasonably easy numerical integration. Upon selecting such a weight function, the ordinary RIS
estimator is defined by
\begin{equation}\label{eq:IS_est}
  \ln\hat{Z}(\beta) = \ln M - \ln\!\left[\frac{1}{n}
  \sum_{i=1}^n m(x_i)e^{\beta U(x_i)}\right].
\end{equation}
This estimator is strongly consistent because $\bE\{m(X)e^{\beta U(X)}\}=M/Z(\beta)$,
where the expectation is with respect to (w.r.t.) $p(\cdot)$,
and so, $\ln Z(\beta)=\ln M- \ln[\bE\{m(X)e^{\beta U(X)}\}]$, which implies
that $\ln\hat{Z}(\beta)\to\ln Z(\beta)$ as $n\to\infty$ almost surely by the strong
law of large numbers.

Since the MSE of such estimators decays at rate $1/n$, we define the
\emph{asymptotic MSE constant} as the limit
$\lim_{n\to\infty}n\cdot\MSE\{\ln\hat{Z}(\beta)\}$,
where 
\begin{equation}
\MSE\{\ln\hat{Z}(\beta)\}\dfn\bE\{[\ln\hat{Z}(\beta)-\ln Z(\beta)]^2\},
\end{equation}
whenever it exists.
For the ordinary RIS estimator \eqref{eq:IS_est}
with a given weight function $m$, a direct
second-moment calculation (see Appendix~C) yields
\begin{equation}\label{eq:mse1}
  V_1(m) \dfn \lim_{n\to\infty}n\cdot\MSE\left\{\ln\hat{Z}(\beta)\right\} = Q_1(m) - 1,
\end{equation}
where $Q_1(m)$ is the \emph{second-moment ratio}
defined by
\begin{equation}\label{eq:Q1}
  Q_1(m) \dfn
  \frac{Z(\beta)\cdot\int_\calX m^2(x)e^{\beta U(x)}\mbox{d}x}
  {\Bigl[\int_\calX m(x)\mbox{d}x\Bigr]^2}.
\end{equation}
The Cauchy--Schwarz inequality implies that $Q_1(m)\ge 1$, with equality iff $m(x)\propto e^{-\beta U(x)}$.
The weight function $m^*(x)=e^{-\beta U(x)}$ will therefore be referred to as
the \emph{ideal weight function}: It achieves zero MSE, but is
infeasible since its normalizing constant
$M^*=\int_{\calX}m^*(x)\mbox{d}x=Z(\beta)$, which is postulated
to be unavailable in closed form in the first place (otherwise, RIS is not
needed anyway).
The RIS design problem is therefore to choose $m$ so as to minimize $Q_1(m)$
from a family of tractable functions, say, a parametric family, such as
$m(x)=e^{-(x-\mu)^2/(2\sigma^2)}$, where $\mu$ and $\sigma^2$ are parameters
and then $M$ is well known to have the simple closed-form
expression, $\sqrt{2\pi\sigma^2}$, provided that $\calX$ is the entire real line. Intuitively,
among all members of our family of tractable functions, we seek the one that
best approximates the ideal weight function $m^*(x)$.

\subsection{Moving on to Groups of Size $k$}

In~\eqref{eq:IS_est}, each sample $x_i$ is weighed separately. This paper
focuses on the following question:
\emph{can one reduce the MSE by processing the
samples in groups of size $k> 1$ rather than weighing each one separately?}

At first glance, the reader may speculate that the answer is necessarily negative: The samples are
statistically independent, so no rearrangement or regrouping can create
new statistical information about $Z(\beta)$. Whatever can be learned
from $x_1,\ldots,x_n$ individually can also be learned from them
collectively, and vice versa. 

While this intuition is conceptually correct as far as the
information content is concerned, it nevertheless overlooks the simple fact that
joint processing of several samples introduces significantly more degrees of
freedom and allows a substantially richer class of weight
functions that potentially better
approximate the corresponding ideal group weight $e^{-\beta[U(x_1)+\cdots+U(x_k)]}$.
The same principle underlies the elements of information theory. For a
memoryless source or channel, successive symbols are independent, yet
block coding over $k$ symbols provably achieves rates, distortions,
and error exponents strictly unattainable on a symbol-by-symbol basis.

The extension of the basic RIS estimator to its grouped version (GRIS) is
conceptually straightforward. Instead of considering $x_1,\ldots,x_n$ as $n$
separate samples, let us view them as $n/k$
vector samples of dimension $k$, generated by grouping, and proceed in the same manner as above.
Specifically, given $n$ i.i.d.\ samples from $p(\cdot)$, 
and given a fixed positive integer $k$ that divides $n$,
partition the data into $n/k$
non-overlapping blocks, or groups, 
\begin{equation}
\bx_j=(x_{jk+1},\ldots,x_{jk+k}),~~
j=0,1,\ldots,\frac{n}{k}-1, 
\end{equation}
and let 
\begin{equation}
U_k(\bx_j)\dfn\sum_{i=1}^k U(x_{jk+i})
\end{equation}
be the
\emph{group energy}. Apply a joint weight $m:\calX^k\to(0,\infty)$ whose
normalizing constant is now $M=\int_{\calX^k}m(\bx)\mbox{d}\bx$, which is assumed calculable
in closed form or by relatively easy numerical integration. The \emph{GRIS estimator} is
defined as
\begin{equation}\label{eq:groupest}
  \ln\hat{Z}^{\rm nol}_k(\beta)
  = \frac{1}{k}\left[\ln M - \ln\left(\frac{k}{n}
    \sum_{j=0}^{n/k-1} m(\bx_j)e^{\beta U_k(\bx_j)}\right)\right],
\end{equation}
where the subscript ``nol'' stands for ``non-overlapping'' in order to
distinguish this estimator from other variants of GRIS that allow overlaps (to
be discussed in the sequel).
As before, since $\bE\{m(\bX)e^{\beta U_k(\bX)}\}=M/Z^k(\beta)$, the
expression on the square brackets of eq.\ \eqref{eq:groupest} (without the
factor $\frac{1}{k}$) estimates $\ln Z^k(\beta)$, and so, upon dividing by
$k$, we obtain
an estimate of $\ln Z(\beta)$.

Let $W_j = m(\bx_j)e^{\beta U_k(\bx_j)}$ and
$\bE\{W_j\} = M/Z^k(\beta)$.
The \emph{group second-moment ratio} is
\begin{equation}\label{eq:Qk}
  Q_k(m) \dfn
  \frac{Z^k(\beta)\cdot\int_{\calX^k}m^2(\bx)e^{\beta
U_k(\bx)}\mbox{d}\bx}
  {\Bigl[\int_{\calX^k}m(\bx)\mbox{d}\bx\Bigr]^2}
  = 1 + \frac{\Var\{W_1\}}{[\bE\{W_1\}]^2},
\end{equation}
where the second equality is verified in Appendix~C.
A direct second-moment calculation (see again Appendix~C)
yields
\begin{equation}\label{eq:mse_derivation}
  V_k^{\rm nol}(m) \dfn \lim_{n\to\infty}n\cdot\MSE\left\{\ln\hat{Z}^{\rm nol}_k(\beta)\right\}
  = \frac{Q_k(m)-1}{k}.
\end{equation}

Again, $Q_k(m)\ge 1$, with equality iff $m(\bx)=m^*(\bx)\propto e^{-\beta U_k(\bx)}$.
GRIS improves on the $k=1$ baseline RIS estimator whenever  
$[Q_k(m)-1]/k < Q_1(m_0)-1$ for some tractable joint weight $m$
and the best tractable single-sample weight $m_0$.
The trade-off is transparent: with $n/k$ groups, instead of $n$, the
variance of the sample mean is $k$ times larger (fewer independent
terms), but $Q_k(m)-1$ may decrease fast enough with $k$ to more than
compensate.

\subsection{General Guidelines for the Choice of Weight Functions}

Conceptually, the simplest joint weight functions have the product-form,
$m(\bx)=\prod_{i=1}^k m_0(x_i)$ for some $m_0:\calX\to(0,\infty)$: they assign independent weights
to each group component and thus introduce no coupling.
One might hope that since the components of $\bx$ are i.i.d., such weight
functions are optimal, or at least very good.
The following
proposition shows this hope is illusory. In fact, product-form weight
functions are {\em strictly} worse than their single-sample counterparts unless
they both yield zero MSE using the ideal weight function.

\begin{proposition}\label{prop:product}
Let $m(\bx)=\prod_{i=1}^k m_0(x_i)$ be any product-form weight function.
Then $Q_k(m)=[Q_1(m_0)]^k$ and
  \begin{equation}\label{eq:product_bound}
    \frac{Q_k(m)-1}{k} = \frac{[Q_1(m_0)]^k-1}{k}\ge Q_1(m_0)-1,
  \end{equation}
with equality iff $Q_1(m_0)=1$.
\end{proposition}

\begin{proof}
Since $m(\bx)=\prod_i m_0(x_i)$ and $e^{\beta U_k(\bx)}=\prod_i e^{\beta U(x_i)}$,
the integrals in~\eqref{eq:Qk} factorize and the relation $Q_k(m)=[Q_1(m_0)]^k$
is readily obtained. Now,
\begin{eqnarray}
\frac{Q_k(m)-1}{k}&=&\frac{[Q_1(m_0)]^k-1}{k}\nonumber\\
&=&\frac{Q_1(m_0)-1}{k}\cdot\sum_{i=0}^{k-1}[Q_1(m_0)]^i\nonumber\\
&\ge&\frac{Q_1(m_0)-1}{k}\cdot\sum_{i=0}^{k-1} 1^i\nonumber\\
&=&Q_1(m_0)-1.
\end{eqnarray}
\end{proof}

Proposition \ref{prop:product} provides our first guideline in the choice of
group-weight functions: For any
improvement from grouping, it is necessary to introduce genuine intra-group coupling,
as the joint weight must \emph{not} factor over individual samples.
This is a {\em negative} guideline - it tells us about one kind of functions to be
ruled out in our quest for good weight functions.

The next proposition suggests a {\em positive} guideline for selecting good weight
functions. It allows us to confine attention to a much smaller and more structured class of weight
functions without loss of optimality. In particular, it tells us that
weight functions that depend on $\bx$ only via its group energy $U_k(\bx)$ are
sufficient. In other words, we may consider only functions of the form
$m(\bx)=\tilde m(U_k(\bx))$ without worrying that we may compromise
performance. We refer
to this type of weight functions as {\em group-energy} weight functions. As
will be discussed shortly, this
simplifies substantially both the optimization and the implementation, since the outer
function $\tilde{m}$ acts on a scalar variable, independently of $k$.

\begin{proposition}\label{prop:ge_optimal}\label{sec:ge_optimal}
  For any weight function $m(\bx)$, there exists a group-energy
  weight function $\tilde m(U_k(\bx))$ whose MSE is no larger
  than that of $m(\bx)$.
\end{proposition}

The proof below uses the notion of the group \emph{density of states}
\begin{equation}\label{eq:Omega}
  \Omega_k(u) \dfn \frac{d}{du}\mathrm{Vol}\bigl\{\bx\in \calX^k:~U_k(\bx)\le u\bigr\},
\end{equation}
that is, the derivative (or the density) of the volume of the sub-level set with respect to the energy
level $u$. Equivalently, $\Omega_k(u)$ is the $(k-1)$-dimensional
hypersurface area of the level set
\begin{equation}
\calS(u)\dfn\{\bx\in \calX^k: U_k(\bx)=u\}.
\end{equation}
From the probabilistic perspective,
this means that the uniform probability measure across $\calS(u)$ has
density $1/\Omega_k(u)$ with respect to the surface measure on
$\calS(u)$. Assuming $U_k(\bx)\ge 0$ for all $\bx\in\calX^k$, the
following co-area identity holds for any integrable function $f$:
\begin{equation}\label{eq:coarea}
  \int_{\calX^k} f(U_k(\bx))\mbox{d}\bx = \int_0^\infty f(u)\Omega_k(u)\mbox{d}u.
\end{equation}
In particular,
\begin{equation}
[Z(\beta)]^k=\int_{\calX^k}e^{-\beta U_k(\bx)}\mbox{d}\bx=\int_0^\infty e^{-\beta
u}\Omega_k(u)\mbox{d}u,
\end{equation}
which means that $[Z(\beta)]^k$ is the Laplace transform of $\Omega_k(u)$~\cite{Huang1987,Pathria2011},
provided that the variable $\beta$ and the function $Z(\beta)$ are both
extended to the entire complex plane.
Consequently, $\Omega_k(u)$ can be obtained as the inverse Laplace transform
of $[Z(\beta)]^k$. 

\begin{proof}
Define
\begin{equation}\label{eq:mbar}
\tilde m(u) \dfn
\frac{\int_{\calS(u)}m(\bx)\mbox{d}\bx}{\Omega_k(u)},
\end{equation}
which depends on $\bx$ only through $U_k(\bx)=u$ and is therefore
a group-energy weight. We show it has the same normalizing constant
as $m(\bx)$ and a numerator that cannot be larger.
\begin{eqnarray}\label{eq:Mtilde}
\tilde M&=&
\int_{\calX^k}\tilde m(U_k(\bx))\mbox{d}\bx\nonumber\\
&=&\int_0^\infty\tilde m(u)\Omega_k(u)\mbox{d}u\nonumber\\
&=&\int_0^\infty \mbox{d}u\int_{\calS(u)} m(\bx)\mbox{d}\bx\nonumber\\
&=&\int_{\calX^k} m(\bx)\mbox{d}\bx\nonumber\\
&=&M,
\end{eqnarray}
where the second equality is due to eq.\ \eqref{eq:coarea} applied to $f(\cdot)=\tilde{m}(\cdot)$,
and the third equality is due to the definition \eqref{eq:mbar}. Thus,
the denominator $M^2$ of $Q_k(\cdot)$ is unchanged.
Now, since $e^{\beta U_k(\bx)}=e^{\beta u}$ is constant across $\calS(u)$,
we can factor it out of the inner integral:
\begin{eqnarray}\label{eq:numerator_chain}
\int_{\calX^k} m^2(\bx)e^{\beta U_k(\bx)}\mbox{d}\bx
&=&\int_0^\infty e^{\beta u}
\int_{\calS(u)} m^2(\bx)\mbox{d}\bx\mbox{d}u\nonumber\\
&=&\int_0^\infty e^{\beta u}
\frac{\int_{\calS(u)} m^2(\bx)\mbox{d}\bx}{\Omega_k(u)}
\cdot\Omega_k(u)\mbox{d}u\nonumber\\
&\ge&\int_0^\infty e^{\beta u}
\left[\frac{\int_{\calS(u)} m(\bx)\mbox{d}\bx}{\Omega_k(u)}\right]^2
\cdot\Omega_k(u)\mbox{d}u\nonumber\\
&=&\int_0^\infty e^{\beta u}\tilde{m}^2(u)\Omega_k(u)\mbox{d}u,
\end{eqnarray}
where the inequality is Jensen's inequality applied to the convex
function $t\mapsto t^2$.
Since the numerator of $Q_k(\tilde m)$ is no larger than that of
$Q_k(m)$, and both share the same denominator $M^2$,
we conclude $Q_k(\tilde m)\le Q_k(m)$.\qed
\end{proof}

\begin{remark}[Rao--Blackwell connection]
Proposition \ref{prop:ge_optimal} has the flavor of the Rao--Blackwell
theorem \cite{vanderVaart1998}: the group energy $U_k(\bx)$
plays the role of a sufficient statistic, and the proof shows that
conditioning the weight on $U_k(\bx)$ (averaging $m(\bx)$ over the
level set $\calS(U_k(\bx))$) never increases the MSE.
The analogy is not exact --- the classical Rao--Blackwell theorem
applies to unbiased estimators and uses Fisher--Neyman sufficiency
--- but the mechanism is the same: Jensen's inequality applied to
$t\mapsto t^2$ shows that the level-set average $\tilde m(u)$
achieves a smaller numerator in $Q_k$ for the same denominator $M$.
\hfill$\lozenge$
\end{remark}

As mentioned briefly before, Proposition~\ref{prop:ge_optimal} is significant in two respects. First, it
reduces dramatically the space of functions in which one should seek good
weight functions, and thereby simplifies the optimization substantially.
Secondly, it reduces significantly the effort needed in order to calculate the
normalization constant $M$ needed for the estimator. Instead of a
$k$-dimensional integral over $\calX^k$, we simply have a one-dimensional integral
over the positive reals. 
\begin{equation}\label{eq:Mk_1d_body}
M=\int_{\calX^k} \tilde m(U_k(\bx))\mbox{d}\bx
=\int_0^\infty \tilde m(u)\Omega_k(u)\mbox{d}u.
\end{equation}
We have, however, to be able to derive
the density of states $\Omega_k(u)$.
This is a purely geometric quantity whose derivation is
given in Appendix~D; the key facts are the following.
\begin{itemize}
  \item For $k=1$: $\Omega_1(u)$ is given by the co-area formula
    (equation~\eqref{eq:Omega1_coarea} in Appendix~D).
  \item For $k\ge 2$: $\Omega_k = \Omega_1^{*(k)}$,
    the $(k-1)$-fold convolution of $\Omega_1$ with itself
    (equation~\eqref{eq:Omega_conv}).  Computing $\Omega_k$ therefore
    requires at most $k-1$ one-dimensional integrations.
    Alternatively, the $k$-fold convolution can be implemented exactly
    via two Fourier-transform integrations regardless of $k$: compute
    $\hat\Omega_1(\omega)=\int_0^\infty e^{i\omega u}\Omega_1(u)\,\mbox{d}u$,
    raise pointwise to the $k$-th power, and invert.
  \item When $\Omega_1(u)\propto u^{a-1}$ (power-law energies
    $U(x)=|x|^\gamma$), the convolution is in
    closed form via the Gamma convolution theorem
    (equation~\eqref{eq:Omega_k_powerlaw}).
  \item Otherwise (e.g.\ $U(x)=(x^2-1)^2$), $\Omega_k$ is computed
    by numerical convolution --- a one-dimensional operation for any $k$.
  \item For large $k$, the saddlepoint approximation
    (Appendix~D.2) gives a closed-form
    approximation valid for any $U$, with relative error $O(1/k)$.
\end{itemize}

\begin{remark}[Discrete state spaces]\label{rem:discrete}
When $\calX$ is a discrete (finite or countable) set,
the Lebesgue measure (volume)
in~\eqref{eq:Omega} is replaced by the counting measure. The density
of states becomes a \emph{cardinality}:
\begin{equation}\label{eq:Omega_discrete}
  \Omega_k(u)\dfn\#\bigl\{\bx\in \calX^k:~U_k(\bx)=u\bigr\},
\end{equation}
and the partition function identity reads
$Z^k(\beta)=\sum_u e^{-\beta u}\Omega_k(u)$,
where the sum is over the image of $U_k$.
All integrals in the GRIS theory are then replaced by sums, and
the $k$-fold convolution~\eqref{eq:Omega_conv} becomes a discrete
convolution (equivalently, a convolution of the probability generating
function of $U(X)$ under $p$).  The Fourier-transform method applies equally
well via the discrete Fourier transform (DFT): compute the DFT of
$\{\Omega_1(u)\}$, raise pointwise to the $k$-th power, and invert.
All other results --- Proposition~\ref{prop:product} (product weights
worsen MSE), Proposition~\ref{prop:ge_optimal} (group-energy weight functions
are sufficient without loss of optimality), the chi-squared identity~\eqref{eq:chi2_id}, and the
sliding-window theory --- carry over verbatim with integrals replaced
by sums.
\hfill$\lozenge$
\end{remark}

Proposition~\ref{prop:ge_optimal} reduces the search to group-energy
weights $\tilde m(u)$. Within this class, a very natural choice is the
\emph{generalized Gaussian} (GG) family
\begin{equation}\label{eq:GG}
  \calM_k^{\rm GG} \dfn
  \left\{m(\bx)=\exp\!\left(-\frac{[U_k(\bx)]^\alpha}{2s}\right)
    :\;\alpha>1,\;s>0\right\},
\end{equation}
with normalizing constant
\begin{equation}\label{eq:Mk_1d}
  M = \int_0^\infty e^{-u^\alpha/(2s)}\Omega_k(u)\mbox{d}u,
\end{equation}
a one-dimensional integral computable from $\Omega_k$, $\alpha$, and $s$
alone. The shape parameter $\alpha>1$ controls the coupling: values near
$\alpha=1^+$ give a weight that closely approximates
the ideal $e^{-\beta U_k(\bx)}$, while $\alpha=2$ gives a Gaussian
weight in $U_k$. 
The restriction $\alpha>1$ is essential: at $\alpha=1$,
$m(\bx)=e^{-U_k(\bx)/(2s)}=\prod_{i=1}^k e^{-U(x_i)/(2s)}$
is a product weight, which Proposition~\ref{prop:product} shows
worsens the MSE.
For power-law energies $U(x)=|x|^\gamma$, the ideal weight $e^{-\beta u}$
is in the GG family at $\alpha\to 1^+$, $s\to 1/(2\beta)$, so
$\Delta_k\to 0$ as $\alpha\to 1^+$.
For non-power-law energies (e.g.\ the double-well of Example~\ref{ex:doublewell}),
the GG family need not contain the ideal weight, and $\Delta_k>0$ for
all $\alpha>1$; the GG family then serves as a tractable parametric
approximation whose quality depends on how well $e^{-u^\alpha/(2s)}$
approximates $e^{-\beta u}$ near the mode of $p_U^k$.

A broader remark is in order.
Proposition~\ref{prop:ge_optimal} establishes that the search for
an optimal weight function can be confined to the group-energy class
without loss of optimality.  Within this class, however, the
choice of a specific parametric family is inevitably a compromise
between two competing demands: approximation quality (how closely
the family can mimic the ideal weight $e^{-\beta u}$) and
tractability (whether $M$ and $Q_k$ can be evaluated efficiently).
These two demands pull in opposite directions --- the ideal weight
is intractable by definition, and any tractable family necessarily
falls short of it.  A universal theoretical answer to ``which family
is best?'' does not exist and should not be expected: the answer
depends on the energy function $U$, the group size $k$, and the
computational budget available.  The GG family is proposed here as
a principled and flexible choice: it is parametrized by two
interpretable scalar parameters ($\alpha$ and $s$), it contains the
ideal weight as a limiting case, its normalizing constant $M$ is a
one-dimensional integral for any $k$, and it smoothly interpolates
between the near-ideal regime ($\alpha\approx 1^+$) and the Gaussian
regime ($\alpha=2$).  Other tractable families (e.g.\ Gamma-type
weights, mixtures, or spline-based approximations) could be used
in its place, and the theoretical framework developed here ---
particularly the chi-squared representation~\eqref{eq:Deltak_chi2}
--- applies equally to any group-energy family.

\subsection{The Chi-Square Divergence Representation}

We next define
\begin{equation}\label{eq:Deltak_def}
  \Delta_k \dfn \min_{m\in\calM_k^{\rm GG}}\frac{Q_k(m)-1}{k}
\end{equation}
for the optimal asymptotic MSE constant at group size $k$ within the GG family,
and $\Delta_1 = \min_{m_0\in\calM_1^{\rm GG}}V_1(m_0)$ for the $k=1$ baseline.
GRIS improves whenever $\Delta_k<\Delta_1$.

For a group-energy weight $\tilde m(U_k(\bx))$, define the
probability density on $(0,\infty)$:
\begin{equation}\label{eq:tq_def}
  \tilde q(u) \dfn \frac{\tilde m(u)\Omega_k(u)}{M},
\end{equation}
and the density of the group energy $U_k(\bx)$ under $p^{\otimes k}$:
\begin{equation}\label{eq:pUk_def}
  p_U^k(u) \dfn \frac{e^{-\beta u}\Omega_k(u)}{Z^k(\beta)}.
\end{equation}
The \emph{chi-squared divergence} between two densities $q$ and $r$
on $(0,\infty)$ is
\begin{equation}\label{eq:chi2def}
  \chi^2(q\|r) \dfn \int_0^\infty \frac{q^2(u)}{r(u)}\mbox{d}u - 1.
\end{equation}
Substituting~\eqref{eq:tq_def} and~\eqref{eq:pUk_def}
into~\eqref{eq:chi2def} gives the identity
\begin{equation}\label{eq:chi2_id}
  Q_k(\tilde m)-1 = \chi^2\left(\tilde q\|p_U^k\right).
\end{equation}
The $k=1$ case reduces to the chi-squared divergence on $\calX$:
defining $q(x)=m(x)/M$,
\begin{equation}\label{eq:chi2_k1}
  Q_1(m)-1 = \chi^2(q\|p),
\end{equation}
which was already noted in~\cite{LiuFang2015} (p.~516).
This extends to $k>1$ via $\Omega_k$.

The key point is that $\chi^2(\tilde q\|p_U^k)$ lives on the
one-dimensional group-energy space $(0,\infty)$, making $\Delta_k$
a tractable one-dimensional optimization:
\begin{equation}\label{eq:Qk_1d}
  Q_k(\tilde m) - 1
  = \frac{Z^k(\beta)\int_0^\infty\tilde m^2(u)e^{\beta u}\Omega_k(u)\mbox{d}u}
         {\left[\int_0^\infty\tilde m(u)\Omega_k(u)\mbox{d}u\right]^2} - 1,
\end{equation}
and hence
\begin{equation}\label{eq:Deltak_chi2}
  \Delta_k = \frac{1}{k}\,\min_{\tilde q\in\calQ_k}
  \chi^2\!\left(\tilde q\;\|\;p_U^k\right),
\end{equation}
where $\calQ_k$ is the image of $\calM_k^{\rm GG}$ under
the map $\tilde m(u)\mapsto \tilde q(u)=\tilde m(u)\Omega_k(u)/M$.
Grouping improves (i.e., $\Delta_k<\Delta_1$) whenever
$\chi^2(\tilde q^*\|p_U^k)<k\Delta_1$, where $\tilde q^*$
is the minimizer in~\eqref{eq:Deltak_chi2}.

The identity~\eqref{eq:chi2_id} thus reduces $\Delta_k$ to a
tractable one-dimensional optimization over $\alpha>1$ and $s>0$.
Numerical illustrations of these formulas, including exact values of
$\Delta_k$ for $k=1,2,3$ and figures showing the MSE as a function of
$\alpha$ and $k$, are given in Section~\ref{sec:numerics}.

\begin{conjecture}[Monotonicity of $\Delta_k$]\label{conj:monotone}
$\Delta_k \le \Delta_{k-1}$ for all $k\ge 2$, i.e., the optimal
asymptotic MSE constant within the GG family is non-increasing in $k$.
\end{conjecture}
\noindent
This is strongly supported by the numerical evidence in
Table~\ref{tab:deltak} and all examples in Section~\ref{sec:numerics},
but a general proof remains open.
For power-law energies $U(x)=|x|^\gamma$, $p_U^k$ is a Gamma
distribution with known parameters, and the conjecture may be
approachable via explicit representations of $\Delta_k$; we leave
this as an open problem.

\subsection{When is $\Omega_k(u)$ Available but $Z(\beta)$ is Not?}
\label{sec:omega_availability}

A natural question arises: if $\Omega_k(u)$ is available in closed form or by
easy numerical integration, can one not
simply compute $Z(\beta)=\int_0^\infty e^{-\beta u}\Omega_1(u)\mbox{d}u$
directly and avoid the need for RIS or GRIS altogether? In general, yes --- and this
is not a contradiction. Nonetheless, the situation where GRIS is genuinely useful
is one where $\Omega_k(u)$ is available but $Z(\beta)$ is not, and this
can occur in two ways.

First, and most importantly, the \emph{perturbed energy} setting
(Section~\ref{sec:perturbed}): when $U(x)=U^\star(x)+\epsilon V(x)$, the
density of states $\Omega_k^\star(u)$ of the main term $U^\star$ is tractable,
but the density of states of the full energy $U$ is not --- because
the perturbation $V$ may couple coordinates or otherwise destroy the
separable structure that made $\Omega_k^\star$ tractable, rendering
$Z(\beta)=\int_0^\infty e^{-\beta u}\Omega_k(u)\mbox{d}u$ intractable.
In this case one uses a weight based on $U^\star$ alone,
and GRIS provides a consistent estimate of $Z(\beta)$ at a controlled
MSE cost (Section~\ref{sec:perturbed}).

Second, even when $\Omega_k$ and $Z(\beta)$ are both available in principle,
computing $Z(\beta)$ via the Laplace transform $\int_0^\infty e^{-\beta u}\Omega_k(u)\mbox{d}u$
may be no easier than the original problem --- for instance, when
$\Omega_k$ is only available numerically and the Laplace integral requires
high-accuracy evaluation at a specific $\beta$. In such cases, GRIS
with samples from $p$ may be the more convenient route.

\section{Sliding-Window Grouping}
\label{sec:overlap}

We now extend the GRIS scheme to incorporate overlapping groups, and in
particular - cyclic sliding windows. Before proceeding, we
flag an important caveat regarding the theoretical guarantees
we had for NOL groups.
Both Proposition~\ref{prop:product} (product-form
weights worsen the MSE) and Proposition~\ref{prop:ge_optimal} (group-energy
weight functions are sufficient without loss of optimality) were proved
for the NOL estimator $\hat{Z}^{\rm nol}_k$
of Section~\ref{sec:setup}, whose asymptotic MSE constant
$V_k^{\rm nol}(m)$ is defined in~\eqref{eq:mse_derivation}
and depends only on
$Q_k(m)$. For the sliding-window variants, the MSE involves
additional cross-covariance terms, henceforth denoted $R_\ell$ ($\ell$ --
positive integer), and neither
result has been established in this setting --- the arguments do not
carry over directly because they do not control these covariance terms.
In particular, it cannot be ruled out that a product-form weight or a
non-group-energy weight achieves smaller sliding-window MSE than any
group-energy weight in the GG family, by exploiting the covariance
structure in ways that a scalar group-energy weight cannot. We
nevertheless continue to work with group-energy weights from the GG
family throughout this section.
As argued at the end of Section~\ref{sec:ge_optimal}, the choice of
a tractable parametric family within the group-energy class is
always a compromise between approximation quality and computational
feasibility, with no universal theoretical optimum.
The GG family was chosen in Section~\ref{sec:ge_optimal} on those
grounds, and those grounds apply equally here: the GG family is
tractable, flexible, and contains the ideal weight as a limit.
We therefore adopt it for the sliding-window schemes as well,
with the understanding that the resulting estimators are not
claimed to be globally optimal --- only that they provably improve
on the NOL baseline, as shown below.
Importantly, what \emph{can} be said on the
positive side is that the sliding-window estimators with group-energy
weights \emph{provably outperform} the NOL estimator
with group-energy weights, whenever a certain condition
holds --- a result that does not require optimality of group-energy
weights, only that they are used consistently in both schemes.
The theoretical question of whether these restrictions entail any
loss relative to the true optimum of the sliding-window problem remains
open (see the open problem described in Section~\ref{sec:conclusions}).

As mentioned earlier, we study two types of sliding-window grouping schemes -- fixed-weight
sliding windows (FSW) and variable-weight sliding windows (VSW).
Throughout this section, overlapping groups are denoted $\tbx_i$ (with
a tilde, to distinguish them from the NOL groups $\bx_j$ of
Section~\ref{sec:setup}).

\subsection{Fixed-Weight Sliding Window Grouping}

The \emph{fixed-weight sliding-window} estimator uses $n$ overlapping groups
with a cyclic wrap-around. In other words given the data $x_1,\ldots,x_n$, our
`samples' are now
\begin{equation}
\tbx_i=(x_i,\ldots,x_{[(i+k-2)~\mbox{\tiny mod}~n]+1}), 
\end{equation}
where
$\ell~\mbox{mod}~n$ designates $\ell$ modulo $n$, namely, $\ell-n\cdot\lfloor
\ell/n\rfloor$. The estimator is given by
\begin{equation}\label{eq:sl_est}
  \ln\hat{Z}^{\rm fsw}_k(\beta)
  = \frac{1}{k}\!\left[\ln M - \ln\!\left(\frac{1}{n}
    \sum_{i=1}^n m(\tbx_i)e^{\beta U_k(\tbx_i)}\right)\right].
\end{equation}
Let $W_i=m(\tbX_i)e^{\beta U_k(\tbX_i)}$ and define the overlap covariances
\begin{equation}\label{eq:sigma_ell}
  R_\ell\dfn\Cov\{W_1,W_{\ell+1}\}\;\text{ for }\ell=0,1,\ldots,k-1,
\end{equation}
For $1\le \ell \le k-1$, groups $\tbx_1$ and $\tbx_{\ell+1}$ share $k-\ell$
samples, and $R_\ell=0$ for all $\ell\ge k$ since groups separated by
$k$ or more steps share no samples. A direct second-moment calculation
(Appendix~C) gives
\begin{eqnarray}\label{eq:mse_sl}
V_k^{\rm fsw}(m) &\dfn&
\lim_{n\to\infty} n\cdot\MSE\left\{\ln\hat{Z}^{\rm
fsw}_k(\beta)\right\}\nonumber\\
&=&\frac{R_0+2\cdot\sum_{\ell=1}^{k-1}R_\ell}
               {k^2\mu_W^2}\nonumber\\
&=&\frac{Q_k(m)-1}{k^2}\!\left(1+2\cdot\sum_{\ell=1}^{k-1}\rho_\ell\right),
\end{eqnarray}
where $\mu_W=\bE\{W_1\}$ and
$\rho_\ell\dfn R_\ell/R_0$, $\ell=1,\ldots,k-1$, are the overlap correlation coefficients.
The FSW estimator outperforms the NOL estimator if and only if
\begin{equation}\label{eq:slide_cond}
  \sum_{\ell=1}^{k-1}\rho_\ell <\frac{k-1}{2}.
\end{equation}
For $k=2$, this reduces to $\rho_1<1/2$.

We next derive the above covariances.
For $\ell<k$, groups $\tbx_1=(x_1,\ldots,x_k)$ and
$\tbx_{\ell+1}=(x_{\ell+1},\ldots,x_{\ell+k})$
share samples $x_{\ell+1},\ldots,x_k$. Their group energies decompose as
$U_k(\tbx_1)=A_\ell+B_\ell$ and $U_k(\tbx_{\ell+1})=B_\ell+C_\ell$, where
\begin{align*}
  A_\ell &= \textstyle\sum_{i=1}^\ell U(x_i) \quad\text{(unshared, first group)},\\
  B_\ell &= \textstyle\sum_{i=\ell+1}^k U(x_i) \quad\text{(shared)},\\
  C_\ell &= \textstyle\sum_{i=k+1}^{k+\ell} U(x_i) \quad\text{(unshared, second group)},
\end{align*}
all three mutually independent.
For a group-energy weight $m(\bx)=\tilde m(U_k(\bx))$, define the
\emph{reweighted kernel} as $\tilde w(u)\dfn\tilde m(u)e^{\beta u}$. Then:
\begin{equation}\label{eq:sigma_decomp}
\bE\{W_1 W_{\ell+1}\}
=\bE\left\{\bE\{\tilde w(A_\ell+B_\ell)\mid B_\ell\}\cdot
    \bE\{\tilde w(B_\ell+C_\ell)\mid B_\ell\}\right\},
\end{equation}
where the outer expectation is over $B_\ell$.
Define
\begin{equation}\label{eq:g_ell}
  g_\ell(v) \dfn \int_0^\infty \tilde w(u+v)e^{-\beta u}\Omega_\ell(u)\mbox{d}u
  = e^{\beta v}\int_0^\infty \tilde m(u+v)\Omega_\ell(u)\mbox{d}u,
\end{equation}
so that $\bE\{\tilde w(A_\ell+v)\mid B_\ell=v\} = 
\bE\{\tilde w(v+C_\ell)\mid B_\ell=v\} = 
g_\ell(v)/Z^\ell(\beta)$,
where the factor $e^{\beta u}$ from $\tilde w(u+v)$ and the factor
$e^{-\beta u}$ from the averaging density of $A_\ell$ cancel,
leaving only $e^{\beta v}$ outside the integral.
Substituting into~\eqref{eq:sigma_decomp} and 
averaging over
$B_\ell$ with density $e^{-\beta v}\Omega_{k-\ell}(v)/Z^{k-\ell}(\beta)$, we obtain
\begin{eqnarray}\label{eq:cross_cov_exact}
\bE\{W_1 W_{\ell+1}\}
& =&\frac{1}{Z^{k+\ell}(\beta)}
\int_0^\infty [g_\ell(v)]^2e^{-\beta v}\Omega_{k-\ell}(v)\mbox{d}v\nonumber\\
&=&\frac{1}{Z^{k+\ell}(\beta)}\int_0^\infty \left[\int_0^\infty \tilde{m}(u+v)\Omega_\ell(u)\mbox{d}u\right]^2
e^{\beta v}\Omega_{k-\ell}(v)\mbox{d}v,
\end{eqnarray}
and therefore, using $\mu_W = M/Z^k(\beta)$:
\begin{equation}\label{eq:R_ell_exact}
R_\ell = \frac{1}{Z^{k+\ell}(\beta)}\int_0^\infty \left[\int_0^\infty \tilde{m}(u+v)\Omega_\ell(u)\mbox{d}u\right]^2
e^{\beta v}\Omega_{k-\ell}(v)\mbox{d}v
-\frac{M^2}{Z^{2k}(\beta)},
\end{equation}
so $R_\ell$ requires two nested one-dimensional integrations.
For the GG weight $\tilde m(u)=e^{-u^\alpha/(2s)}$, equation~\eqref{eq:g_ell}
becomes
\begin{equation}\label{eq:g_ell_GG}
  g_\ell(v) = e^{\beta v}\int_0^\infty e^{-(u+v)^\alpha/(2s)}\Omega_\ell(u)\mbox{d}u,
\end{equation}
which for $U(x)=|x|$, $\Omega_\ell(u)=2^\ell u^{\ell-1}/(\ell-1)!$,
is computable by one-dimensional numerical integration for each $v$.
These expressions are used in the numerical examples of Section~\ref{sec:numerics}.

\subsection{Variable-Weight Sliding Window Grouping}
\label{sec:periodic}

We now allow the weight function vary across group positions with the hope
of reducing the MSE  further.
To fix ideas, consider first the simplest case of $k=2$, where we now allow ourselves to
alternate between two weight functions: $m_1$ for odd-indexed
pairs, $(x_{2j-1},x_{2j})$, and $m_2$ for even-indexed pairs,
$(x_{2j},x_{2j+1})$, $j=1,2,\ldots$. The estimator is
\begin{equation}\label{eq:per_est_k2}
  \ln\hat{Z}^{\rm vsw}_2(\beta)
  = \tfrac{1}{2}\left[\ln\tfrac{M_1+M_2}{2}
    - \ln\left(\tfrac{1}{n}\sum_{i=1}^n W_i\right)\right],
\end{equation}
where $W_{2j-1}=m_1(\tbx_{2j-1})e^{\beta U_2(\tbx_{2j-1})}$ and
$W_{2j}=m_2(\tbx_{2j})e^{\beta U_2(\tbx_{2j})}$.
For a generic pair $\tbx\sim p^{\otimes 2}$, denote
$W_l\dfn m_l(\tbx)e^{\beta U_2(\tbx)}$ for $l=1,2$,
so that $\bE\{W_l\}=M_l/Z^2(\beta)$; odd-indexed groups have 
$W_i\overset{d}{=}W_1$ and even-indexed groups have $W_i\overset{d}{=}W_2$.
To see why this is a consistent estimator,
note that for each $l\in\{1,2\}$,
\begin{equation}\label{eq:per_unbiased}
  \bE\left\{m_l(\tbx)e^{\beta U_2(\tbx)}\right\}
  = \int_{\calX^2} m_l(\tbx)\frac{e^{-\beta U_2(\tbx)}}{Z^2(\beta)}\cdot e^{\beta U_2(\tbx)}\mbox{d}\tbx
  = \frac{M_l}{Z^2(\beta)}.
\end{equation}
The sample mean therefore converges a.s.\ to
$(M_1+M_2)/[2Z^2(\beta)]$, and
so $\ln\hat{Z}^{\rm vsw}_2(\beta)\to\ln Z(\beta)$, as required.

For the cyclostationary process $\{W_i\}_{i\ge 1}$, adjacent groups share
one sample and $\bar{R}_\ell=0$ for $\ell\ge 2$,
so $\sum_\ell|\bar{R}_\ell|$ is finite. A direct second-moment calculation
(Appendix~C, Section~C.3) gives
\begin{equation}\label{eq:mse_k2_periodic}
  V_2^{\rm vsw}(m_1,m_2) \dfn
  \lim_{n\to\infty}n\cdot\MSE\{\ln\hat{Z}^{\rm vsw}_2(\beta) \}=
  \frac{\tfrac{1}{2}[\Var\{W_1\}+\Var\{W_2\}] + 2\bar{R}_1}
       {4\bar\mu_W^2},
\end{equation}
where $\bar\mu_W=(M_1+M_2)/[2Z^2(\beta)]$ and
$\bar{R}_1=\Cov\{W_{2j-1},W_{2j}\}$ (same for all $j$ by cyclostationarity)
is the cross-covariance at lag~$1$ between adjacent groups of different
types, which share one sample.
The VSW improves on the FSW
($k=2$ case) iff
\begin{equation}\label{eq:periodic_cond_k2}
  \tfrac{1}{2}[\Var\{W_1\}+\Var\{W_2\}] + 2\bar{R}_1
  \;<\; \Var\{W\} + 2R_1,
\end{equation}
and this is achievable in principle, as we now argue intuitively.
With two free parameters $s_1$ and $s_2$, the VSW scheme has
more degrees of freedom than the FSW (one
parameter $s$), so there is potential room for improvement beyond
what is achievable with a single weight. Whether this potential is
always realized, and under what conditions, is an open question.
In practice, the optimal $(s_1,s_2)$ is found by numerically
minimizing~\eqref{eq:mse_k2_periodic}, and the numerical evidence
consistently shows a strict improvement, as illustrated in the
next section.

We now derive an integral representation of $\bar{R}_1$ for $k=2$.
Adjacent groups $\tbx_{2j-1}=(x_{2j-1},x_{2j})$ and $\tbx_{2j}=(x_{2j},x_{2j+1})$
share the sample $x_{2j}$. Write $V=U(x_{2j})$ for the energy of the
shared sample, and $A=U(x_{2j-1})$, $B=U(x_{2j+1})$ for the energies
of the two unshared samples, all three i.i.d.\ $\sim p_U$.
For a group-energy weight $\tilde m_l$ with reweighted kernel
$\tilde w_l(u)=\tilde m_l(u)e^{\beta u}$, define
\begin{equation}\label{eq:gl_k2}
  g_l(v) \dfn e^{\beta v}\int_0^\infty \tilde m_l(u+v)\,\Omega_1(u)\,\mbox{d}u,
  \qquad l=1,2,
\end{equation}
so that $\bE\{\tilde w_l(A+v)\} = g_l(v)/Z(\beta)$
(the factors $e^{\beta u}$ from $\tilde w_l$ and $e^{-\beta u}$ from the
density of $A\sim p_U$ cancel, as in~\eqref{eq:g_ell}).
Conditioning on $V$ and using the independence of $A$ and $B$ given $V$:
\begin{align}\label{eq:R1_k2}
  \bE\{W_{2j-1}W_{2j}\}
  &= \bE\left\{\bE\{\tilde w_1(A+V)\mid V\}\cdot\bE\{\tilde w_2(V+B)\mid
V\}\right\}
  \nonumber\\
  &= \frac{1}{Z^3(\beta)}\int_0^\infty g_1(v)\,g_2(v)e^{-\beta v}\Omega_1(v)\mbox{d}v,
\end{align}
and therefore
\begin{equation}\label{eq:barR1_exact}
  \bar{R}_1
  = \frac{1}{Z^3(\beta)}\int_0^\infty g_1(v)g_2(v)e^{-\beta v}\Omega_1(v)\mbox{d}v
    - \frac{M_1 M_2}{Z^4(\beta)}.
\end{equation}
This exact formula requires one one-dimensional integration per $g_l$ evaluation
(the inner integral in~\eqref{eq:gl_k2}) and one outer integration in $v$.
For the GG weight $\tilde m_l(u)=e^{-u^\alpha/(2s_l)}$:
\begin{equation}\label{eq:gl_GG}
  g_l(v) = e^{\beta v}\int_0^\infty e^{-(u+v)^\alpha/(2s_l)}\,\Omega_1(u)\,\mbox{d}u.
\end{equation}

This is in the spirit of antithetic variates~\cite{Hammersley1956},
but applied to the weight functions rather than the samples:
the samples remain i.i.d.\ from $p$, and the MSE reduction
is achieved by alternating weight functions with different scale parameters.
Numerical results are given in Section~\ref{sec:numerics}.

Moving on to a general $k$, we now
cycle through $k$ weight functions
$m_1,\ldots,m_k$, using $m_{1+((i-1)\bmod k)}$ for group $\tbx_i$,
with $W_l\dfn m_l(\tbx)e^{\beta U_k(\tbx)}$ for $l=1,\ldots,k$
and a generic group $\tbx\sim p^{\otimes k}$. The same consistency argument extends:
$\bE\{W_l\}=M_l/Z^k(\beta)$ for each $l$, so the sample mean
converges to $\bar\mu_W=\frac{1}{k}\sum_{l=1}^{k}M_l/Z^k(\beta)$
and the estimator converges to $\ln Z(\beta)$.
For the time-averaged cross-covariance $\bar{R}_\ell$ at lag $\ell$,
the same conditioning argument generalizes~\eqref{eq:barR1_exact}.
Groups at lag $\ell$ share $k-\ell$ consecutive samples with total
energy $B_\ell$; the unshared portions have energies
$A_\ell$ and $C_\ell$. For weight functions
$m_{j}$ and $m_{j'}$ with $j'=1+((j-1+\ell)\bmod k)$, define
\begin{equation}\label{eq:gl_general}
  \phi_j^\ell(v) \dfn e^{\beta v}\int_0^\infty
  \tilde m_j(u+v)\Omega_\ell(u)\mbox{d}u, \qquad j=1,\ldots,k,
\end{equation}
so that $\bE\{\tilde w_j(A_\ell+v)\}=\phi_j^\ell(v)/Z^\ell(\beta)$
(since $A_\ell\perp B_\ell$, the conditional expectation given $B_\ell=v$
equals the marginal expectation at fixed $v$).
Then, averaging over pairs $(j,j')$ separated by lag $\ell$:
\begin{equation}\label{eq:barRell_exact}
  \bar{R}_\ell
  = \frac{1}{k}\sum_{j=1}^{k}
    \left[\frac{1}{Z^{k+\ell}(\beta)}\int_0^\infty
    \phi_j^\ell(v)\,\phi_{j'}^\ell(v)\,
    e^{-\beta v}\,\Omega_{k-\ell}(v)\,\mbox{d}v
    - \frac{M_j\,M_{j'}}{Z^{2k}(\beta)}\right],
\end{equation}
a formula requiring at most two nested one-dimensional integrations
for each term.
The sign and magnitude of $\bar{R}_\ell$ depend on the choice of
weight parameters $(s_1,\ldots,s_k)$ and are computed
exactly via~\eqref{eq:barRell_exact}.

The asymptotic MSE constant (Appendix~C, Section~C.3) is:
\begin{equation}\label{eq:mse_periodic}
  V_k^{\rm vsw}(m_1,\ldots,m_k) \dfn
  \lim_{n\to\infty}n\cdot\MSE\left\{\ln\hat{Z}^{\rm vsw}_k(\beta)\right\}
  = \frac{1}{k^2\bar\mu_W^2}\!\left[
    \frac{1}{k}\sum_{l=1}^{k}\Var\{W_l\}
    + 2\sum_{\ell=1}^{k-1}\bar{R}_\ell\right],
\end{equation}
with parameters $(s_1,\ldots,s_k)$ chosen by minimizing this
expression numerically.
The non-overlapping formula corresponds to $\bar R_\ell=0$ (groups share
no samples, regardless of the weights), and the sliding-window formula
to the case $m_l=m$ for all $l$, giving $\bar R_\ell=R_\ell$
(see Appendix~C).

The VSW improves on the FSW whenever
\begin{equation}\label{eq:periodic_cond}
  \frac{1}{k}\sum_{l=1}^{k}\Var\{W_l\} + 2\sum_{\ell=1}^{k-1}\bar{R}_\ell
  \;<\;
  \Var\{W\} + 2\sum_{\ell=1}^{k-1} R_\ell.
\end{equation}
Numerical illustrations for $k=2$ and $k=3$ are given in Section~\ref{sec:numerics}.

\section{Numerical Examples}
\label{sec:numerics}

First, a few words are in order concerning the choice of examples.
As discussed in Section~\ref{sec:omega_availability}, the scenarios
where $\Omega_k$ is tractable but $Z(\beta)$ is not are the perturbed
energy setting (Section~\ref{sec:perturbed}) and genuinely
high-dimensional problems. The three examples below do not fall into
either category: they are one-dimensional with both $\Omega_k$ and
$Z(\beta)$ tractable. They are deliberately included as \emph{toy examples} serving
as \emph{controlled experiments} only. Since $Z(\beta)$ is available in
closed-form exactly, in these examples, we can report exact MSE values, computed by one-dimensional
integration rather than noisy Monte Carlo estimates. This allows
clean verification of the theory and an unambiguous comparison of the three
tiers of GRIS (non-overlapping, FSW, and VSW).
In realistic applications, $Z(\beta)$ would not be available, and the
GRIS estimator would be applied exactly as described here, with the
one-dimensional integrals for $M$ and $Q_k(m)$ evaluated from the
known or computable $\Omega_k$. We also note that the toy examples
can be interpreted as the $\epsilon\to 0$ limit of the perturbed
Hamiltonian setting of Section~\ref{sec:perturbed}: as the perturbation
vanishes, $Z(\beta)$ becomes tractable and the unperturbed examples
are recovered. The GRIS method itself works for every value of
$\epsilon$, and the MSE grows continuously from its unperturbed value
as $\epsilon$ increases.

\begin{example}\label{ex:1prime}
Let $U(x)=|x|$, $\calX=\reals$, and $\beta=1$. Then $Z(1)=2$ and
$p(x)=\frac{1}{2}e^{-|x|}$.
The group density of states is $\Omega_k(u)=2^k u^{k-1}/(k-1)!$
(see eq.\ \eqref{eq:Omega_k_exp}).
For the Gaussian weight $\tilde m(u)=e^{-u^2/(2s)}$, the normalizing
constant is
\begin{equation}\label{eq:Mk_gamma}
M=\frac{2^k}{(k-1)!}\int_0^\infty u^{k-1}e^{-u^2/(2s)}\mbox{d}u
=\frac{2^{k-1}}{(k-1)!}\cdot(2s)^{k/2}\Gamma\!\left(\frac{k}{2}\right),
\end{equation}
where $\Gamma(t)\dfn\int_0^\infty u^{t-1}e^{-u}\mbox{d}u$, and the
exact formula for $\Delta_k$ (to be optimized over $s>0$) is
\begin{equation}\label{eq:Deltak_explicit}
\frac{Q_k(m)-1}{k}
=\frac{1}{k}\left[
\frac{(k-1)!\displaystyle\int_0^\infty u^{k-1}e^{u-u^2/s}\mbox{d}u}
{2^{k-2}\cdot(2s)^k[\Gamma(k/2)]^2}
-1\right],
\end{equation}
using the closed form for the denominator already derived in~\eqref{eq:Mk_gamma}.
Throughout, we use $m_s(\bx)=e^{-[U_k(\bx)]^2/(2s)}$ with $s$ optimized
per scheme.
The natural external comparator for GRIS is the best $k=1$ RIS estimator
within the GG family with $\alpha$ also optimized.
As Table~\ref{tab:deltak} shows, this gives $\Delta_1^{\rm opt}\approx 0$
because the ideal weight $e^{-\beta U(x)}$ is in the GG family at
$\alpha\to 1^+$, $s\to 1/(2\beta)$ --- but achieving this requires
setting $s=Z(\beta)/2$, which is the unknown we are trying to estimate.
The practically accessible regime is therefore fixed $\alpha=2$
(Gaussian group-energy weight, $s$ optimized), where the GRIS gains of
$33$--$50\%$ at $k=2,3$ are the meaningful figure; and in the perturbed
Hamiltonian setting of Section~\ref{sec:perturbed}, where $Z(\beta)$ is
genuinely intractable, this is the only option. 

\medskip\noindent{\em Non-overlapping grouping.}
Table~\ref{tab:deltak} reports on results of $\Delta_k$ for $k=1,\ldots,10$:
the MSE falls from $0.0807$ at $k=1$ to $0.0144$ at $k=10$ ($82\%$
reduction), with diminishing marginal returns as $k$ grows.
For the FSW and VSW comparisons below we fix $\alpha=2$ (Gaussian
weight, $s$ optimized), so that the sliding-window gains are assessed
on top of the same NOL baseline.

\begin{table}[h!]
\centering
\begin{tabular}{cc}
\toprule
$k$ & $\Delta_k$\\
\midrule
  1  & 0.08074 \\
  2  & 0.05411 \\
  3  & 0.04041 \\
  4  & 0.03216 \\
  5  & 0.02669 \\
  8  & 0.01763 \\
  10 & 0.01437 \\
\bottomrule
\end{tabular}
\caption{Calculated values of $\Delta_k = \min_{s>0}[Q_k(m_s)-1]/k$
for $U(x)=|x|$, $\beta=1$, Gaussian weight
$m_s(\bx)=e^{-[U_k(\bx)]^2/(2s)}$, $s$ optimized via
one-dimensional minimization of~\eqref{eq:Deltak_explicit}.}
\label{tab:deltak}
\end{table}

Figure~\ref{graph1} shows $V_k^{\rm nol}(m_s)$ vs.\ shape exponent $\alpha$
for $k=1,2,3$: all curves approach zero as $\alpha\to 1^+$ (ideal weight),
and grouping strictly reduces the MSE for every $\alpha>1$.
The Gaussian choice $\alpha=2$ is used throughout.

\medskip\noindent{\em Fixed-weight sliding window (FSW).}
Condition~\eqref{eq:slide_cond} is satisfied at $k=2$ ($\rho_1=0.292<0.5$)
and $k=3$ ($\rho_1+\rho_2=0.591<1.0$), so the FSW improves on NOL
(Table~\ref{tab:sliding}): the cumulative gain rises from $33\%$ to $47\%$
at $k=2$, and from $50\%$ to $64\%$ at $k=3$.

\medskip\noindent{\em Variable-weight sliding window (VSW).}
The VSW adds a further $30$--$40\%$ on top of the FSW (Table~\ref{tab:sliding}).
At $k=2$, $(s_1^*,s_2^*)=(0.816,3.081)$ gives a $62.8\%$ cumulative
gain --- already matching the $k=3$ FSW at two-thirds the cost.
At $k=3$, $(s_1^*,s_2^*,s_3^*)=(1.491,1.491,4.484)$ gives
$V_3^{\rm vsw}=0.0178$, a $77.9\%$ cumulative gain.

\begin{table}[h!]
\centering
\begin{tabular}{ccccc}
\toprule
$k$ & Scheme & $s^*$ / $(s_{\rm lo},s_{\rm hi})$ & $V_k$ & \% gain vs.\ $k=1$ \\
\midrule
1 & NOL & 1.411 & 0.0807 & --- \\
2 & NOL & 2.379 & 0.0541 & 33.0\% \\
2 & FSW & 2.387 & 0.0428 & 46.9\% \\
2 & VSW ($s_{\rm lo}{=}0.816$, $s_{\rm hi}{=}3.081$) & --- & 0.0300 & 62.8\% \\
3 & NOL & 3.365 & 0.0404 & 50.0\% \\
3 & FSW & 3.373 & 0.0294 & 63.6\% \\
3 & VSW ($s_1{=}s_2{=}1.491$, $s_3{=}4.484$) & --- & 0.0178 & 77.9\% \\
\bottomrule
\end{tabular}
\caption{Exact asymptotic MSE constant $V_k$ for Example~\ref{ex:1prime}, Gaussian
  group-energy weight $m=e^{-[U_k(\bx)]^2/(2s)}$, $s$ optimized per row.
  NOL, FSW, and
  VSW with optimal antithetic pair $(s_{\rm lo},s_{\rm hi})$.
  All values by exact one-dimensional integration.}
\label{tab:sliding}
\end{table}

\medskip\noindent{\em Larger group sizes.}
Figure~\ref{fig:mse_vs_k} extends the comparison to $k=1,\ldots,8$.
NOL and FSW points are exact global optima; VSW is exact for $k\le 4$
(solid) and an upper bound for $k\ge 5$ (dashed, antithetic
$(s_{\rm lo},\ldots,s_{\rm lo},s_{\rm hi})$ pattern).
All three curves decrease steadily and stay roughly evenly spaced;
at $k=8$ the cumulative gains are $78.2\%$ (NOL), $85.4\%$ (FSW),
and $92.4\%$ (VSW upper bound).

\begin{figure}[!htbp]
\centering
\includegraphics[width=0.62\textwidth]{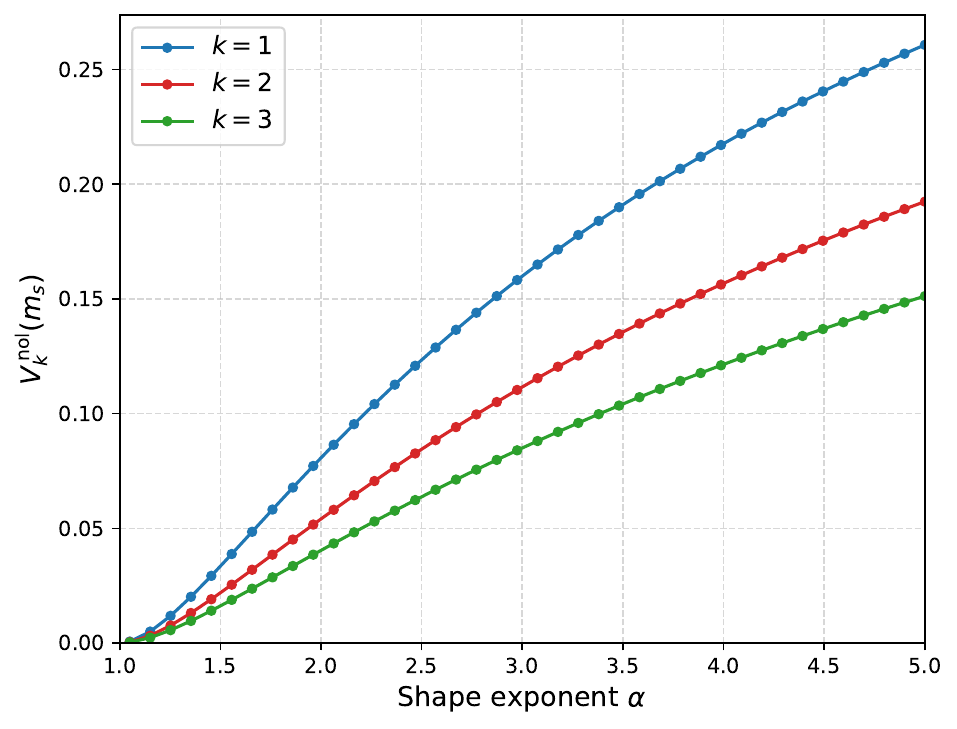}
\caption{Asymptotic MSE constant $V_k^{\rm nol}(m_s)$ vs.\ shape exponent
$\alpha$ for NOL
  grouping, $k=1,2,3$ (Example~\ref{ex:1prime}).}
\label{graph1}

\medskip
\includegraphics[width=0.62\textwidth]{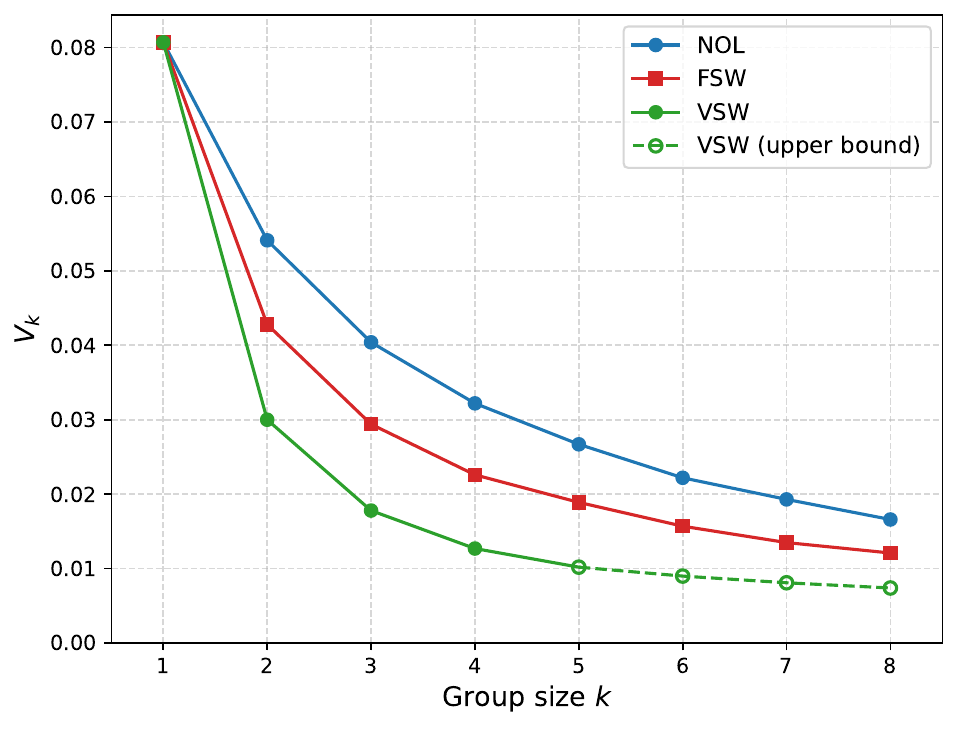}
\caption{Asymptotic MSE constant $V_k$ vs.\ group size $k$: NOL (blue),
  FSW (red), VSW (green; solid $k\le4$ exact, dashed $k\ge5$ upper bound).}
\label{fig:mse_vs_k}
\end{figure}

\end{example}

\begin{example}\label{ex:stretched_exp}
Let $U(x)=|x|^{3/2}$, $\calX=\reals$, and $\beta=1$.
A change of variables gives $Z(\beta)=\tfrac{4}{3}\Gamma(\tfrac{2}{3})\approx 1.805$,
and by eq.\ \eqref{eq:Omega_k_stretch} in Appendix~D:
  \begin{equation}\label{eq:omega_stretched}
    \Omega_k(u) = \frac{\bigl[\tfrac{4}{3}\Gamma\!\left(\tfrac{2}{3}\right)\bigr]^k}
                       {\Gamma(2k/3)}\,u^{2k/3-1},\quad u>0.
  \end{equation}
We use $m(\bx)=e^{-[U_k(\bx)]^2/(2s)}$ with $s$ (or $(s_1,\ldots,s_k)$
for VSW) optimized via~\eqref{eq:Qk_1d},\eqref{eq:sigma_decomp}.
Table~\ref{tab:ex2} summarizes the results.
  \begin{table}[h!]\centering
  \begin{tabular}{cccccc}
  \toprule
  $k$ & Scheme & $\Delta_k$ & \% gain & $\sum_{\ell=1}^{k-1}\rho_\ell$ & SW cond. \\
  \midrule
  1 & ordinary RIS & 0.06395 & ---   & ---   & --- \\
  2 & NOL & 0.04633 & 27.5\%& ---   & --- \\
  2 & FSW & 0.03788 & 40.8\%& 0.318 & $0.318<0.5$ \checkmark \\
  2 & VSW ($s_1{=}0.432$, $s_2{=}2.493$) & 0.02129 & 66.7\% & --- & --- \\
  3 & NOL & 0.03607 & 43.6\%& ---   & --- \\
  3 & FSW & 0.02715 & 57.5\%& 0.629 & $0.629<1.0$ \checkmark \\
  3 & VSW ($s_1{=}s_2{=}0.817$, $s_3{=}3.608$) & 0.01205 & 81.2\% & --- & --- \\
  \bottomrule
  \end{tabular}
  \caption{Exact asymptotic MSE constant $V_k$ for Example~2
    ($U(x)=|x|^{3/2}$, Gaussian weight $\alpha=2$, $s$ optimized per row).}
  \label{tab:ex2}
  \end{table}
The FSW satisfies condition~\eqref{eq:slide_cond} at both $k=2$ and $k=3$,
improving on NOL by $13$--$14\%$.
The VSW again adds roughly $30\%$ on top of the FSW:
at $k=2$, $(s_1^*,s_2^*)=(0.432,2.493)$ gives $V_2^{\rm vsw}=0.02129$
($66.7\%$ cumulative gain, already below the $k=3$ FSW);
at $k=3$, $(s_1^*,s_2^*,s_3^*)=(0.817,0.817,3.608)$ gives
$V_3^{\rm vsw}=0.01205$ ($81.2\%$ cumulative gain).

\medskip\noindent{\em Interpretation.}
Freeing $\alpha$ (optimizing over $\alpha>1$) gives $\alpha^*\approx 1.05$
for all $k$, collapsing $V_k$ below $0.00054$ --- grouping then
adds only marginally, since the $k=1$ weight with $\alpha\approx 1$
already closely approximates the ideal $e^{-\beta U(x)}$ and
$\Delta_1\approx 0$. This illustrates the fundamental principle:
\emph{grouping helps when and because the $k=1$ weight class cannot
approximate the ideal weight well}.
\end{example}

\begin{example}\label{ex:doublewell}\label{sec:ex3}
Let $U(x)=(x^2-1)^2$, $\calX=\reals$, and $\beta=1$. The Boltzmann
distribution $p(x)\propto e^{-(x^2-1)^2}$ is \emph{bimodal},
with peaks at $x=\pm 1$ and a barrier at $x=0$.
The partition function $Z(\beta)\approx 1.974$ is computed numerically.
Unlike the first two examples, $U$ is not a power of $|x|$, so
$\Omega_k$ has no closed form.

$\Omega_k$ is computed via~\eqref{eq:Omega1_coarea} and the
convolution~\eqref{eq:Omega_conv}, each step one-dimensional.
The weight $m(\bx)=e^{-[U_k(\bx)]^2/(2s)}$ is used throughout;
the $k=1$ baseline gives $\Delta_1\approx 0.02674$ at $s^*\approx 0.597$.
Overlap covariances $R_\ell$ are evaluated by nested one-dimensional
integration in the original coordinates.

  \begin{table}[h!]\centering
  \small
  \begin{tabular}{clccccc}
  \toprule
  $k$ & Scheme & Optimal $s$ & $\Delta_k$ & \% gain & $\sum_{\ell=1}^{k-1}\rho_\ell$ & SW cond. \\
  \midrule
  1 & Ordinary RIS & $0.597$ & 0.02674 & ---    & ---   & --- \\
  2 & NOL & $1.088$ & 0.02090 & 21.8\%  & ---   & --- \\
  2 & FSW & $1.088$ & 0.01732 & 35.2\% & $\rho_1=0.327$ & $0.327<0.5$ \checkmark \\
  2 & VSW & $(0.251,2.037)$ & 0.00574 & 78.5\% & --- & --- \\
  3 & NOL & $1.535$ & 0.01604 & 40.0\%  & ---   & --- \\
  3 & FSW & $1.535$ & 0.01284 & 52.0\%  & $\rho_1{+}\rho_2=0.702$ & $0.702<1.0$ \checkmark \\
  3 & VSW (upper bound) & $(0.35,0.35,2.9)$ & 0.00305 & 88.6\% & --- & --- \\
  \bottomrule
  \end{tabular}
  \caption{Exact asymptotic MSE constant $V_k$ for Example~\ref{ex:doublewell}
    ($U(x)=(x^2-1)^2$, Gaussian weight $\alpha=2$, $s$ optimized per row).
    The $k=2$ VSW value is the exact joint optimum over $(s_1,s_2)$,
    computed by direct multi-dimensional numerical integration (the
    density-of-states $\Omega_k$ has no closed form for this example,
    so~\eqref{eq:R_ell_exact} was evaluated directly in the original
    sample coordinates rather than via $\Omega_k$).  The $k=3$ VSW
    value is a rigorous \emph{upper bound}: it uses the restricted
    antithetic configuration $(s_1,s_1,s_2)$ found by a coarse grid
    search rather than a full joint optimization over $(s_1,s_2,s_3)$,
    so the true VSW optimum at $k=3$ is no larger than $0.00305$.}
  \label{tab:ex3}
  \end{table}

The VSW again yields a large further gain (Table~\ref{tab:ex3}):
at $k=2$, $(s_1^*,s_2^*)=(0.251,2.037)$ gives $V_2^{\rm vsw}=0.00574$
($66.9\%$ below the FSW, $78.5\%$ cumulative gain);
at $k=3$, the antithetic pattern $(s_1,s_1,s_2)=(0.35,0.35,2.9)$
gives $V_3^{\rm vsw}\le 0.00305$ (upper bound), a $76\%$ improvement
over the FSW.
The NOL gains ($22$--$40\%$) are smaller than in the other examples,
consistent with the smaller $\Delta_1=0.0267$.
Condition~\eqref{eq:slide_cond} is satisfied at both $k=2$
($\rho_1=0.327<0.5$) and $k=3$ ($\rho_1+\rho_2=0.702<1.0$).
Table~\ref{tab:all_examples} collects all results.
\end{example}

\begin{table}[h!]
\centering
\small
\renewcommand{\arraystretch}{1.15}
\begin{tabular}{llcccccc}
\toprule
 & & \multicolumn{6}{c}{$V_k$ (\% gain vs.\ $k=1$)} \\
\cmidrule(lr){3-8}
Ex. & $U(x)$ & $k{=}1$ & $k{=}2$ NOL & $k{=}2$ FSW & $k{=}2$ VSW & $k{=}3$ FSW & $k{=}3$ VSW \\
\midrule
1 & $|x|$
  & 0.0807
  & 0.0541 (33)
  & 0.0428 (47)
  & 0.0300 (63)
  & 0.0294 (64)
  & 0.0178 (78) \\
2 & $|x|^{3/2}$
  & 0.0640
  & 0.0463 (28)
  & 0.0379 (41)
  & 0.0213 (67)
  & 0.0272 (58)
  & 0.0121 (81) \\
3 & $(x^2{-}1)^2$
  & 0.0267
  & 0.0209 (22)
  & 0.0173 (35)
  & 0.0057 (79)
  & 0.0128 (52)
  & $\le$0.0031 ($\ge$89) \\
\bottomrule
\end{tabular}
\caption{Summary of $V_k$ across all three examples
  (Gaussian weight, $s$ optimized per cell;
  percentage gains vs.\ $k=1$ in parentheses).
  The $k=3$ VSW entry for Example~3 is a rigorous upper bound;
  see Table~\ref{tab:ex3}.}
\label{tab:all_examples}
\end{table}

\section{Perturbed Energy}
\label{sec:perturbed}

The three examples above
all have a partition function that is available in closed form, or at least
one that lends itself to one-dimensional numerical integration. This raises the question of whether
GRIS is useful when $Z(\beta)$ is genuinely unavailable in that sense.
If the weight depends on $\bx$ only through $U_k(\bx)$, then
$\Omega_k$ is needed for $M=\int_0^\infty\tilde m(u)\Omega_k(u)\mbox{d}u$.
But the identity $\int_0^\infty\Omega_k(u)e^{-\beta u}\mbox{d}u=Z^k(\beta)$
shows that $Z(\beta)$ is then a one-dimensional integral of the same
type and hence equally tractable. Whenever $\Omega_k$ is
computable by the methods of Appendix~D, so is
$Z(\beta)$ --- the two are inseparable.

Genuine intractability arises when
\begin{equation}\label{eq:perturbed_energy}
  U(x) = U^\star(x) + \epsilon V(x),
\end{equation}
where $U^\star(x)$ is a tractable main term and $V(x)$ is an arbitrary
perturbation with $|V(x)|\le 1$ uniformly. For $\epsilon>0$,
$Z(\beta)$ has no closed form. Crucially, if $V(x)$ involves
complicated cross-terms between coordinates --- for example,
$V(x)=\sin(x_1)\sin(x_2)$ in $\reals^2$ --- then $Z(\beta)$ cannot
even be reduced to a one-dimensional integral: the cross-term couples
all coordinates and the density of states of the \emph{full} energy
is genuinely intractable.

The key idea, in such scenarios, is to base the weight on the tractable term
only, i.e., on $U_k^\star(\bx)=\sum_{i=1}^k U^\star(x_i)$
rather than the full group energy $U_k(\bx)$.
We use, for example, the GG weight 
\begin{equation}
m(\bx)=\exp\left\{-\frac{1}{2s}\left[\sum_{i=1}^k
U^\star(x_i)\right]^\alpha\right\} 
\end{equation}
with parameters $\alpha>1$ and $s>0$, and normalizing constant
\begin{equation}
M=\int_0^\infty e^{-u^\alpha/(2s)}\Omega_k^\star(u)\mbox{d}u 
\end{equation}
involving only the tractable density of states $\Omega_k^\star$ of $U^\star$.
Both $\alpha$ and $s$ are optimized jointly.

As mentioned before, the restriction $\alpha>1$ is essential: at $\alpha=1$ the weight
factors as $\prod_i e^{-U^\star(x_i)/(2s)}$, a product form that,
similarly as in Proposition~\ref{prop:product}, can be shown to worsen the MSE.
The natural choice is $\alpha$ close to $1^+$: the ideal (zero-MSE)
weight for the unperturbed problem is $\exp\left\{-\beta\sum_{i=1}^k
U^\star(x_i)\right\}$,
which lies at the product-form boundary $\alpha=1$ (with $s=1/(2\beta)$).
The closer $\alpha$ is to $1$, the better the weight approximates
the ideal while remaining non-product.

Since $|V(x)|\le 1$, we have $e^{\beta\epsilon\sum_i V(x_i)}\le e^{k\beta\epsilon}$
and $(Z(\beta)/Z^\star(\beta))^k \le e^{k\beta\epsilon}$, so
$Q_k(m)\le e^{2k\beta\epsilon}Q_k^\star(m)$ and hence:
\begin{equation}\label{eq:perturbed_bound}
  V_k^{\rm nol}(m) = \frac{Q_k(m)-1}{k} \le
  e^{2k\beta\epsilon}\,V_k^{\rm nol,\star}(m)
  + \frac{e^{2k\beta\epsilon}-1}{k}
  \;\lesssim\; V_k^{\rm nol,\star}(m) + 2\beta\epsilon,
\end{equation}
where $V_k^{\rm nol,\star}(m) = [Q_k^{\star}(m)-1]/k$ is the asymptotic MSE constant
for the unperturbed energy $U^\star$, and the approximation holds for
small $k\beta\epsilon$. The full GRIS machinery --- $M$, the chi-squared
formula, and the sliding-window covariances --- is tractable because
it involves only $\Omega_k^{\star}$.

A sufficient condition for the bound~\eqref{eq:perturbed_bound}
at group size $k$ to be smaller than at $k=1$ is
\begin{equation}\label{eq:perturbed_crit}
  e^{2k\beta\epsilon}\,V_k^{\rm nol,\star}(m)
  + \frac{e^{2k\beta\epsilon}-1}{k}
  <
  e^{2\beta\epsilon}\,V_1^{\rm nol,\star}(m_0)
  + (e^{2\beta\epsilon}-1),
\end{equation}
where $V_1^{\rm nol,\star}(m_0)=[Q_1^{\star}(m_0)-1]$ is the unperturbed
$k=1$ MSE constant.
For small $k\beta\epsilon$ this reduces to $V_k^{\star}<V_1^{\star}$,
which holds whenever grouping helps in the unperturbed problem.
For larger $\epsilon$, the exponential factors penalize larger $k$:
writing $\rho=e^{2\beta\epsilon(k-1)}>1$, condition~\eqref{eq:perturbed_crit}
becomes $\rho\,V_k^{\star}+(\rho-1)/k < V_1^{\star}$, which fails once the
perturbation overhead $(\rho-1)/k$ dominates.
In particular, for a near-ideal weight ($V_k^{\star}\approx 0$),
the condition fails as soon as $e^{2k\beta\epsilon}-1>k(e^{2\beta\epsilon}-1)$,
i.e.\ for any $\epsilon>0$ and $k\ge 2$ --- consistent with the
numerical findings in Tables~\ref{tab:perturbed}--\ref{tab:perturbed_2d}.

Below we demonstrate this in two sub-examples.

\paragraph{Sub-example 1: $U(x)=|x|+\epsilon\cos(x)$ in $\calX=\reals$ ($\beta=1$).}
Take $U^\star(x)=|x|$ and $V(x)=\cos(x)$.  For $\epsilon>0$, $Z(\beta)$
has no closed form.  We use the weight based on $\sum_i|x_i|$
with $\Omega_k^{\star}(u)=2^k u^{k-1}/(k-1)!$ from Example~\ref{ex:1prime}.
Table~\ref{tab:perturbed} shows results for $k=1,2$.
All values are exact: for $k=1$ by one-dimensional integration
(with jointly optimized $\alpha^*,s^*$ at $\epsilon>0$); for $k=2$
by two-dimensional numerical integration with $s$ fixed at the
unperturbed optimum ($s^*=2.379$).

\begin{table}[h!]\centering
\begin{tabular}{cccccc}
\toprule
$\epsilon$ & $\alpha^*$ & $\Delta_1$ (opt.\ $\alpha$) & $\Delta_2$ (opt.\ $\alpha$) & $\Delta_1$ ($\alpha{=}2$) & $\Delta_2$ ($\alpha{=}2$) \\
\midrule
0.00 & 1.001 & $0$ & $0$ & 0.08074 & 0.05411 \\
0.05 & 1.020 & 0.00014 & $4.95\times10^{-5}$ & 0.07942 & 0.05232 \\
0.10 & 1.044 & 0.00061 & 0.00057 & 0.07973 & 0.05215 \\
0.20 & 1.104 & 0.00249 & 0.00305 & 0.08534 & 0.05668 \\
0.50 & 1.441 & 0.01378 & 0.02350 & 0.14507 & 0.11420 \\
\bottomrule
\end{tabular}
\caption{Asymptotic MSE constant $V_k$ for $U(x)=|x|+\epsilon\cos(x)$, $\beta=1$.
  Weight $m(\bx)=e^{-[\sum_i|x_i|]^\alpha/(2s)}$ based on $U^\star(x)=|x|$.
  For $\epsilon>0$, $Z(\beta)$ has no closed form.
  All values exact by numerical integration ($k=1$ by 1D integration;
  $k=2$ exact by 2D integration (both opt-$\alpha$ and Gaussian)).
  The shared $\alpha^*$ column gives the optimal exponent for $k=1$;
  the optimal $\alpha^*$ for $k=2$ is within $0.003$ of these values.
  At $\epsilon=0$ the opt-$\alpha$ entries are exactly $0$ (the
  ideal weight $e^{-\beta U^\star}$, attained as $\alpha\to 1^+$,
  achieves zero MSE for the unperturbed problem).}
\label{tab:perturbed}
\end{table}

As $\epsilon$ grows, $\alpha^*$ moves away from $1$ reflecting the
need to balance approximation quality against perturbation cost.
The table reveals a sharp distinction between the two weight regimes.
With $\alpha=2$ (Gaussian), grouping helps throughout: $\Delta_2 < \Delta_1$
for all $\epsilon$, with a $21$--$35\%$ gain.
With the optimal $\alpha$ (near $1^+$), the picture reverses for moderate
$\epsilon$: at $\epsilon=0.20$ and $\epsilon=0.50$,
$\Delta_2^{\rm opt} > \Delta_1^{\rm opt}$, meaning grouping hurts.
This is consistent with the corrected bound~\eqref{eq:perturbed_bound}:
when $\alpha\approx 1$, the $k=1$ estimator already nearly eliminates
the MSE ($\Delta_1^{\rm opt}\approx 0$), and the factor
$e^{2k\beta\epsilon}$ in the bound amplifies the perturbation cost
proportionally to $k$, so $k=2$ incurs twice the perturbation penalty
for a baseline that is already near zero.

\paragraph{Sub-example 2: $U(x)=|x_1|+|x_2|+\epsilon\sin(x_1)\sin(x_2)$
in $\calX=\reals^2$ ($\beta=1$).}
Here $U^\star(x)=|x_1|+|x_2|$ and $V(x)=\sin(x_1)\sin(x_2)$ is a bounded
cross-term.  For $\epsilon>0$, $Z(\beta)=\int_{\reals^2}
e^{-|x_1|-|x_2|-\epsilon\sin(x_1)\sin(x_2)}\,\mbox{d}x$ has no closed form
and cannot be reduced to a one-dimensional integral, since the cross-term
couples both coordinates. The unperturbed density of states
$\Omega_1^{\star}(u)=4u$ (sum of two $\mathrm{Exp}(1)$ variables).

For $k=1$ (a single 2D sample), the GG weight is
$m(x)=e^{-[U^\star(x)]^\alpha/(2s)}=e^{-(|x_1|+|x_2|)^\alpha/(2s)}$
with $M=\int_0^\infty e^{-u^\alpha/(2s)}\cdot 4u\mbox{d}u$ tractable.
For $k=2$ (a group of two 2D samples $x^{(1)},x^{(2)}\in\reals^2$),
$\sum_{d=1}^2(|x^{(1)}_d|+|x^{(2)}_d|)$ is the sum
of four i.i.d.\ $\mathrm{Exp}(1)$ variables, giving tractable
$\mathrm{Gamma}(4,1)$ density of states.

\begin{table}[h!]\centering
\begin{tabular}{cccccc}
\toprule
$\epsilon$ & $\alpha^*$ & $\Delta_1$ (opt.\ $\alpha$) & $\Delta_2$ (opt.\ $\alpha$) & $\Delta_1$ ($\alpha{=}2$) & $\Delta_2$ ($\alpha{=}2$) \\
\midrule
0.00 & 1.001 & $0$ & $0$ & 0.1075 & 0.0642 \\
0.20 & 1.001 & 0.00387 & 0.00544 & 0.1156 & 0.0706 \\
0.50 & 1.001 & 0.03822 & 0.04002 & 0.1589 & 0.1137 \\
1.00 & 1.001 & 0.16914 & 0.18562 & 0.3230 & 0.2938 \\
\bottomrule
\end{tabular}
\caption{Asymptotic MSE constant $V_k$ for
  $U(x)=|x_1|+|x_2|+\epsilon\sin(x_1)\sin(x_2)$ in $\reals^2$, $\beta=1$.
  Weight $m(x)=e^{-(|x_1|+|x_2|)^\alpha/(2s)}$ based on $U^\star(x)=|x_1|+|x_2|$.
  At $\epsilon=0$: exact values ($k=1$ by 2D integration; $k=2$ by 1D integration).
  At $\epsilon>0$: $k=1$ exact by 2D integration; $k=2$ opt-$\alpha$ and Gaussian
  by Monte Carlo ($1.5\times10^6$ samples).
  Unlike Sub-example~1, $\alpha^*\approx 1.001$ for both $k=1$ and $k=2$
  at all $\epsilon$.  The $\epsilon=0$ opt-$\alpha$ entries are exactly $0$
  (ideal weight achieves zero MSE for the unperturbed problem).}
\label{tab:perturbed_2d}
\end{table}

With $\alpha=2$, grouping again helps throughout ($10$--$40\%$ gain).
With $\alpha\approx 1.001$, grouping again hurts for all $\epsilon>0$.

\medskip\noindent\textbf{Takeaway.}
Both sub-examples exhibit the same dichotomy, which the
bound~\eqref{eq:perturbed_bound} explains.
The bound decomposes the perturbed MSE constant into two terms:
(i)~$e^{2k\beta\epsilon}V_k^{\rm nol,\star}(m)$, the amplified
unperturbed MSE constant, and (ii)~$(e^{2k\beta\epsilon}-1)/k$, the
perturbation overhead.
When the weight is \emph{far from ideal} for the unperturbed problem
(e.g.\ $\alpha=2$), term~(i) is large at $k=1$ but decreases
substantially with $k$ --- this is exactly the grouping gain ---
while term~(ii) is negligible for moderate $k\beta\epsilon$.
The net effect is that grouping helps.
When the weight is \emph{near ideal} ($\alpha\approx 1^+$),
term~(i) is already near zero at $k=1$, so grouping reduces it
only marginally; meanwhile term~(ii) grows with $k$, making each
additional group sample more expensive under perturbation.
The net effect is that grouping hurts.

The practical implication is clear.
In a genuinely intractable problem one cannot drive $\alpha$ close to
$1$ without knowing $Z(\beta)$ (since the ideal weight is
$e^{-\beta U}$, which requires $Z$).
The accessible regime is therefore that of moderate $\alpha$
(e.g.\ $\alpha=2$), where grouping consistently reduces the MSE
even under perturbation.
The deeper lesson is that \emph{the benefit of grouping is governed
by the gap between the weight class and the ideal weight}: when that
gap is large, grouping closes it; when it is already small, grouping
adds perturbation cost without compensating gain.

\section{Conclusions and Open Problems}
\label{sec:conclusions}

This paper introduced and analyzed GRIS, a method for reducing
the asymptotic MSE constant of partition function estimation by
applying joint weights to groups of samples rather than weighting
each sample separately.  Three successive tiers of improvement were
derived and verified: NOL grouping, FSW grouping, and VSW grouping
with antithetic weight alternation, achieving reductions of
$20$--$65\%$ in $V_k$ across three examples with qualitatively
different energy functions.

One direction that we find particularly interesting remains open.
The NOL and FSW schemes are the two extremes of a family
parametrized by the step size $d$ by which the sliding window
advances: NOL corresponds to $d=k$ (no overlap) and FSW to $d=1$
(overlap of $k-1$ samples).  Intermediate step sizes $d\in\{2,\ldots,k-1\}$
with $d\mid k$ interpolate between the two extremes, and
each step size $d$ decomposes the group naturally into
non-overlapping chunks of size $d$, with adjacent groups sharing
exactly $k/d - 1$ chunks.
This spectrum reveals a fundamental tradeoff.  On the one hand,
larger $d$ (less overlap) reduces the number of free weight
parameters available for optimization, since all samples within
a chunk share the same weight, but it also reduces the
cross-covariance terms $R_\ell$ that inflate the MSE.  On the
other hand, a smaller step size $d$ increases overlap and
introduces more cross-covariance, but allows finer-grained
weight variation. The optimal step size $d^*(k)$ balancing
these two effects is unknown.
A key structural question concerns the sufficiency of
\emph{chunk-energy} weight functions $\tilde{m}(U_d(A_1),
U_d(A_2), \ldots)$, where $A_j$ denotes the $j$-th chunk of $d$
consecutive samples and $U_d(A_j)=\sum_{i\in A_j}U(x_i)$.
We conjecture that restricting to weights depending on $\bx$
only through the chunk-energy vector is sufficient without loss
of optimality for the step-$d$ sliding window, by a
Rao--Blackwell argument with the chunk-energy vector as
sufficient statistic.  If true, the normalizing constant $M$
would reduce from a $k$-dimensional integral to a
$(k/d)$-dimensional integral involving the $d$-fold convolution
$\Omega_d$~\eqref{eq:Omega}, interpolating between the
one-dimensional case of NOL ($d=k$, group-energy weights,
Proposition~\ref{prop:ge_optimal}) and the full
$k$-dimensional case of FSW ($d=1$, individual-energy weights,
for which sufficiency holds but tractability is lost).
Establishing this conjecture and characterizing the optimal
$d^*(k)$ remain open.

\medskip
A second direction concerns \emph{multi-dimensional sample arrays}.
Throughout this paper the $n$ samples are indexed by a single
integer $i=1,\ldots,n$.  In many physical applications ---
spin lattices, crystal arrays, spatial random fields --- the samples
are naturally indexed by a two-dimensional (or higher) array
$\{x_{ij}:\,i,j=1,\ldots,\sqrt{n}\}$ of i.i.d.\ draws from $p$.
For NOL grouping, the extension is conceptually straightforward:
partition the array into non-overlapping rectangular blocks of
size $k_1\times k_2$, apply a joint weight to each block, and
observe that Propositions~\ref{prop:product} and~\ref{prop:ge_optimal}
carry over verbatim because the group energy $\sum_{(i,j)\in B}U(x_{ij})$
depends only on the block size $|B|=k_1k_2$, not on its shape or
orientation (the samples being i.i.d.).

The problem becomes genuinely richer in two respects that have no
one-dimensional counterpart.

\emph{(i) Two-dimensional VSW antithetic structure.}
In 1D, VSW cycles through a sequence of $k$ weight functions along
a single axis, creating antithetic variation in one direction.
In a 2D array, one can tile the plane with a $k_1\times k_2$
weight matrix $\{m_{i'j'}\}$, creating antithetic variation
\emph{simultaneously in both directions}.  The cross-covariances
$\bar{R}_{(\ell_1,\ell_2)}$ are now indexed by 2D lag vectors
$(\ell_1,\ell_2)$, and the optimal weight tile can exploit
cancellations between horizontal and vertical antithetic pairs that
are simply unavailable in 1D.  For example, in a $2\times 2$
tile with one low-scale weight at position $(1,1)$ and a high-scale
weight at positions $(1,2)$, $(2,1)$, $(2,2)$, adjacent groups
along both axes carry opposite-scale weights, potentially reducing
covariance in both directions at once.  The geometry of the optimal
tile --- and whether purely antithetic structure (as in 1D) remains
optimal --- is an open question with no 1D analogue.

\emph{(ii) Non-i.i.d.\ samples: the Markov case.}
A separate and independent extension --- already meaningful in 1D ---
arises when the samples $x_1,x_2,\ldots$ are \emph{not} i.i.d.\ but
form a Markov chain with stationary distribution $p$.  This is the
typical output of an MCMC sampler.  The NOL estimator remains
consistent, but its asymptotic variance now includes contributions
from all pairwise correlations $\Cov\{W_i,W_j\}$ across
\emph{non-overlapping} groups as well, not just within them.
More importantly, Proposition~\ref{prop:ge_optimal} relied on the
i.i.d.\ structure: the conditional distribution of $\bx$ given
$U_k(\bx)=u$ is uniform on the level set $\calS(u)$ only when the
components are i.i.d., and this uniformity is what makes the
Rao--Blackwell step work.  For a Markov chain, the joint distribution
of a group $(x_i,\ldots,x_{i+k-1})$ is not a product measure, so the
level sets $\calS(u)$ are no longer traversed uniformly.
Whether the group-energy sufficiency result holds in this setting
is unclear; the Rao--Blackwell argument relied on the product
structure of $p$, which is no longer available.  Extending GRIS
theory to the Markov setting --- characterizing when group-energy
weights remain approximately sufficient, and how to account for
inter-group correlations in the MSE formula --- is an open problem
even in 1D.

In the 2D array setting with spatial interactions (e.g., the 2D Ising
model), both extensions are active simultaneously: the samples are
neither i.i.d.\ nor one-dimensional, and the GRIS theory would need to
address both the richer group geometry of~(i) and the non-product
joint distribution of~(ii) at once.

\section*{Appendix A -- Comments on Forward Importance Sampling}
\renewcommand{\theequation}{A.\arabic{equation}}
    \setcounter{equation}{0}

The baseline estimator studied throughout this paper is a RIS estimator:
samples are drawn from the Boltzmann distribution $p$ and
reweighted by a tractable function $m$.  A natural
question is whether the grouping idea extends to \emph{forward} (direct) IS.
In FIS one draws $n$ i.i.d.\ samples $x_1,\ldots,x_n$ from a
proposal $q$ and estimates $Z(\beta)$ by
\begin{equation}\label{eq:FIS_est}
  \hat{Z}^{\rm fis}(\beta) = \frac{1}{n}\sum_{i=1}^n
    \frac{e^{-\beta U(x_i)}}{q(x_i)},
\end{equation}
which is unbiased since $\bE_q\{e^{-\beta U(x)}/q(x)\}=Z(\beta)$.
The asymptotic MSE constant is $[Q_1^{\rm fwd}(q)-1]$, where
$Q_1^{\rm fwd}(q)=\bE_q\{[e^{-\beta U(x)}/q(x)]^2\}/Z^2(\beta)$
(see, e.g.,~\cite{LiuFang2015,OwenZhou2000,RobertCasella2004}),
minimized when $q(x)\propto e^{-\beta U(x)}$, i.e.\ $q=p$ --- but
then $q$ is the target itself and sampling from it requires knowing
$Z(\beta)$.  The question is whether grouping helps when $q$ is only
approximately optimal.

In grouped forward IS with a joint group-energy proposal
$\tilde q(\bx)=\tilde r(U_k(\bx))/C_k$,
where $C_k=\int_0^\infty \tilde r(u)\Omega_k(u)\mbox{d}u$.
Each group $\bx_j$ is drawn from $\tilde q$ and weighted by
\begin{equation}
  W_j^{\rm fwd} = \frac{e^{-\beta U_k(\bx_j)}}{\tilde q(\bx_j)}
                = C_ke^{-\beta U_k(\bx_j) + [U_k(\bx_j)]^\alpha/(2s)}
                \quad\text{(for }\tilde r(u)=e^{-u^\alpha/(2s)}\text{)},
\end{equation}
with $\bE_{\tilde q}\{W_j^{\rm fwd}\}=Z^k(\beta)$. The asymptotic
MSE constant is $[Q_k^{\rm fwd}(\tilde r)-1]/k$, where
\begin{equation}\label{eq:Qk_fwd}
  Q_k^{\rm fwd}(\tilde r)-1
  = \frac{C_k\cdot\int_0^\infty \mbox{d}u\Omega_k(u)e^{-2\beta u}/\tilde r(u)}
         {Z^{2k}(\beta)} - 1
  = \chi^2\!\left(\tilde r_*\|p_U^k\right),
\end{equation}
an identity structurally identical to~\eqref{eq:chi2_id} for RIS,
where $\tilde r_*$ is the probability measure on $(0,\infty)$ with density
proportional to $\tilde r(u)\Omega_k(u)$ (the forward-IS analogue of $\tilde q$).
The product-form result (Proposition~\ref{prop:product}) carries over
unchanged: product proposals $\tilde q(\bx)=\prod_i q(x_i)$ actually
\emph{worsen} the MSE relative to $k=1$; the gain from non-product
proposals follows the same chi-squared reduction mechanism.

There is a sign reversal in the tail requirement. In RIS, the
integrability of $\tilde m^2(u)e^{\beta u}$ is
guaranteed for any super-exponentially decaying $\tilde m(u)$, which is the
case for all $\alpha>1$ in the GG family. In forward IS, the integrand
$e^{-2\beta u}/\tilde r(u)$ must be integrable, which requires the proposal
$\tilde r$ to dominate $e^{-2\beta u}$ in the tail --- i.e., $\tilde r$
must have \emph{heavier} tails than the squared Boltzmann factor
$e^{-2\beta u}$.  A Gaussian
group-energy proposal $\tilde r(u)=e^{-u^2/(2s)}$ decays much faster
than $e^{-2\beta u}$ for large $u$, making the forward IS MSE much larger
than the RIS MSE at the same $\alpha=2$. Specifically, for $U(x)=|x|$,
$\beta=1$:
\begin{center}
\begin{tabular}{lccc}
\toprule
Method & $k=1$ & $k=2$ & $k=3$ \\
\midrule
RIS, NOL & 0.081 & 0.054 & 0.040 \\
Forward IS, NOL & 2.000 & 2.834 & 3.915 \\
\bottomrule
\end{tabular}
\end{center}
The forward IS MSE \emph{worsens} with $k$ for the Gaussian group-energy
proposal, because the tail mismatch between the proposal and the target
compounds as the group energy grows.  With the right proposal family
(exponential tilt, $\alpha\to 1$), both methods can achieve near-zero MSE ---
but the optimal forward IS proposal is then the Laplace distribution
$q^*(x)\propto e^{-\beta|x|}=p(x)$, i.e., the target itself.  Forward
IS with the optimal proposal reduces to RIS.

\subsection*{A.1 Why the sliding window is available in RIS but not in forward IS}

The sliding window works in RIS for a simple structural reason:
the $n$ samples $x_1,\ldots,x_n$ are drawn from $p$ independently of
the weight function $m$.  The weight is applied as a post-processing step,
and any group $\tbx_i=(x_i,\ldots,x_{i+k-1})$ has the correct $p^{\otimes k}$
marginal distribution regardless of which other groups share its samples.
Unbiasedness holds for every group separately:
\begin{equation}
  \bE_p\!\left\{m(\tbx_i)\,e^{\beta U_k(\tbx_i)}\right\} = \frac{M}{Z^k(\beta)}
  \quad \text{for every } i,
\end{equation}
because the marginal distribution of $G_i$ is $p^{\otimes k}$ no matter
which samples it shares with adjacent groups.  The proposal and the weight
are \emph{decoupled}: samples come from $p$ independently of what $m$ does.

In forward IS this decoupling breaks down.  The group $\tbx_j$ must be drawn
from $\tilde q$ in order for the weight $e^{-\beta U_k(\tbx_j)}/\tilde q(\tbx_j)$ to be
an unbiased estimator of $Z^k(\beta)$.  If instead one draws $x_i$ i.i.d.\ from
the marginal $q_1$ of $\tilde q$ and forms overlapping groups $\tbx_i=
(x_i,\ldots,x_{i+k-1})$, then $\tbx_i\sim q_1^{\otimes k}$, which equals
$\tilde q$ only if $\tilde q$ is a product proposal.  But product proposals actually
\emph{worsen} the MSE relative to $k=1$ (Proposition~\ref{prop:product}).  Thus:
\begin{itemize}
\item A sliding window with a \emph{product} proposal is unbiased but useless.
\item A sliding window with a \emph{non-product} proposal is useful but biased.
\end{itemize}
In one sentence: \emph{in RIS the proposal ($p$) and the weight
($m$) are decoupled, so sharing samples across groups is free; in forward
IS the proposal and the weight are locked together, so sharing samples
breaks unbiasedness.}

This asymmetry is the main reason GRIS is the more
natural and complete theory: whenever samples from $p$ are available,
GRIS is preferable.
Grouped forward IS is nevertheless useful when the $k=1$ proposal $q$
approximates $p$ reasonably but not well: a joint group-energy proposal
$\tilde r(U_k)$ can then better match the joint Boltzmann factor
$e^{-\beta U_k(\bx)}$ across a range of total energies.  The key requirement
is efficient sampling from $\tilde q$: in the group-energy case this
reduces to drawing $U_k$ from $\tilde r(u)\Omega_k(u)/C_k$ and then
sampling $\bx$ uniformly from the level set $\{U_k(\bx)=u\}$, which is tractable
for the energy functions considered here.

\section*{Appendix B -- Comparison with Related Methods}
\renewcommand{\theequation}{B.\arabic{equation}}
    \setcounter{equation}{0}

\paragraph{Methods that also use samples from $p$.}
When samples from $p$ are available, the natural competitors to GRIS
are the standard RIS estimator ($k=1$, the baseline) and bridge
sampling~\cite{MengWong1996}.  Bridge sampling additionally requires a
tractable reference $q$ and estimates $Z(\beta)/Z_q$ from mixed samples
of $p$ and $q$; its MSE is governed by $\chi^2(p\|q)+\chi^2(q\|p)$.
GRIS avoids the need for a reference entirely.  In applications where
bridge sampling is already in use (e.g., Bayesian model comparison),
GRIS can serve as the RIS component of the bridge, further reducing
variance within that framework.

\paragraph{Annealed importance sampling (AIS).}
AIS~\cite{Neal2001} constructs a sequence of intermediate distributions
bridging a tractable reference $p_0$ to the target $p$, and accumulates
incremental importance weights along a Markov chain.  Unlike GRIS, AIS
does not require samples from $p$: it is designed precisely for settings
where direct sampling from $p$ is infeasible.  The two methods therefore
address fundamentally different problems and should not be seen as
competitors.

GRIS is appropriate when samples from $p$ are available (e.g., after
an MCMC chain has converged, or in physical simulations at equilibrium)
but $Z(\beta)$ remains unknown.  In this regime GRIS requires no bridge,
no reference distribution, and no Markov chain mixing: the only inputs
are the samples themselves and a tractable weight family.  AIS, on the
other hand, is the right tool when direct sampling from $p$ is
infeasible and one must instead traverse a path of intermediate
distributions from a tractable reference.

\section*{Appendix C -- Derivations of Asymptotic MSE Constants}
\renewcommand{\theequation}{C.\arabic{equation}}
    \setcounter{equation}{0}

All three MSE formulas are derived from the same second-moment argument.
No asymptotic normality is required; only the variance of the relevant
sample mean is needed.

For any estimator of the form
$\widehat L = \frac{1}{k}[\ln C - \ln\bar W_n]$,
where $C$ is a known constant and $\bar W_n = n^{-1}\sum_{i=1}^n W_i$,
let $\mu_W = \lim_{n\to\infty}\bE[\bar W_n]$ (finite and positive)
so that $L = \frac{1}{k}\ln(C/\mu_W)$ is the true value.
Write $\bar W_n = \mu_W(1+\varepsilon_n)$ where
$\varepsilon_n = (\bar W_n - \mu_W)/\mu_W \to 0$.  Then
\begin{equation}\label{eq:taylor_expand}
  \widehat L - L
  = -\frac{1}{k}\ln(1+\varepsilon_n)
  = -\frac{\varepsilon_n}{k} + O(\varepsilon_n^2).
\end{equation}
Squaring and taking expectations, the higher-order terms contribute
$O(n^{-2})$ to the MSE, giving the master formula:
\begin{equation}\label{eq:mse_master}
  \lim_{n\to\infty} n\cdot\MSE\{\widehat{L}\}
  = \frac{\lim_{n\to\infty} n\cdot\Var\{\bar W_n\}}{k^2\mu_W^2}.
\end{equation}
The three formulas differ only in how $n\cdot\Var\{\bar W_n\}$ is computed.

\paragraph{Finite-sample validity.}
The approximation $n\cdot\MSE\approx V_k$ is accurate when the
higher-order remainder $O(n^{-2})$ is small relative to the leading
term $V_k/n$.  Since the remainder arises from $\bE[\varepsilon_n^2]^2
= [V_k(m)/n]^2 \cdot O(1)$, the relative error of the first-order
approximation is $O(V_k/n)$.  For the numerical examples in
Section~\ref{sec:numerics}, $V_k\le 0.08$ and the approximation is
therefore accurate to within $1\%$ for $n\ge 800$ and to within
$0.1\%$ for $n\ge 8000$, which are modest sample sizes by any
practical standard.  For larger $V_k$ (e.g.\ in the perturbed
setting at large $\epsilon$), larger $n$ may be needed, but the
asymptotic formula remains the correct characterization of the
estimator's quality in the large-$n$ regime.

\paragraph{C.1 Non-overlapping MSE~\eqref{eq:mse_derivation}.}
The $n/k$ group weights $W_j = m(\bx_j)e^{\beta U_k(\bx_j)}$
are i.i.d.\ with mean $\mu_W=M/Z^k(\beta)$ and variance
$\sigma_W^2=\Var\{W_1\}$.  Setting $C=M$:
\begin{equation}
  n\cdot\Var\{\bar W_{n/k}\} = k\,\sigma_W^2,
\end{equation}
so by~\eqref{eq:mse_master}:
\begin{equation}
  V_k^{\rm nol}(m) = \frac{k\sigma_W^2}{k^2\mu_W^2} = \frac{\sigma_W^2}{k\mu_W^2}
  = \frac{Q_k(m)-1}{k},
\end{equation}
using $\sigma_W^2/\mu_W^2 = Q_k(m)-1$.  To verify this identity:
\begin{align}
  \bE_{p^{\otimes k}}[W_1^2]
  &= \int_{\calX^k}[m(\bx)]^2 e^{2\beta U_k(\bx)}
     \cdot\frac{e^{-\beta U_k(\bx)}}{Z^k(\beta)}\,\mbox{d}\bx
   = \frac{\int_{\calX^k}[m(\bx)]^2 e^{\beta U_k(\bx)}\,\mbox{d}\bx}{Z^k(\beta)},
\end{align}
so $\bE[W_1^2]/\mu_W^2 = Z^k(\beta)\int[m]^2e^{\beta U_k}\,\mbox{d}\bx/M^2 = Q_k(m)$,
and $\sigma_W^2/\mu_W^2 = Q_k(m) - 1$. \qed

\paragraph{C.2 Fixed-weight sliding-window MSE~\eqref{eq:mse_sl}.}
Now $W_i = m(\bx_i)e^{\beta U_k(\bx_i)}$ with $\bx_i=(x_i,\ldots,x_{i+k-1})$
and $C=M$.  The sequence $(W_i)$ is strictly stationary
(each $W_i$ depends only on the i.i.d.\ block $(x_i,\ldots,x_{i+k-1})$),
with autocovariances $R_\ell=\Cov\{W_1,W_{1+\ell}\}$.
For $|\ell|\ge k$, the blocks are disjoint and $R_\ell=0$.  Therefore:
\begin{equation}
  n\cdot\Var\{\bar W_n\}
  = \sum_{|\ell|<n}\!\!\!\left(1-\frac{|\ell|}{n}\right)R_\ell
  \;\xrightarrow{n\to\infty}\; R_0 + 2\sum_{\ell=1}^{k-1}R_\ell
\end{equation}
(dominated convergence; the sum is finite since $R_\ell=0$ for $|\ell|\ge k$).
By~\eqref{eq:mse_master} and $R_0=(Q_k(m)-1)\mu_W^2$, $\rho_\ell=R_\ell/R_0$:
\begin{equation}
  V_k^{\rm fsw}(m) = \frac{R_0+2\sum_{\ell=1}^{k-1}R_\ell}{k^2\mu_W^2}
  = \frac{Q_k(m)-1}{k^2}\!\left(1+2\sum_{\ell=1}^{k-1}\rho_\ell\right). \quad\qed
\end{equation}

\paragraph{C.3 Variable-weight sliding window MSE~\eqref{eq:mse_periodic}.}
The weight sequence cycles with period $k$: group $i$ uses weight
$m_{i\bmod k}$, giving $W_i = m_{i\bmod k}(\tbx_i)e^{\beta U_k(\tbx_i)}$.
Set $C = \bar\mu_W Z^k(\beta)$ where $\bar\mu_W = k^{-1}\sum_{l=1}^kM_l/Z^k(\beta)$.
The process $\{W_i\}$ is \emph{cyclostationary} (periodically stationary)
with period $k$: the distribution of $W_i$ depends on $i$ only through
$i\bmod k$, and similarly for all finite-dimensional distributions.
In particular, $\Cov\{W_i,W_{i+\ell}\}$ depends on both $\ell$ and
$i\bmod k$, not on $\ell$ alone. Define the \emph{time-averaged
covariance} at lag $\ell$ as
\begin{equation}
  \bar R_\ell \dfn \frac{1}{k}\sum_{l=1}^{k}
  \Cov\{W_i,W_{i+\ell}\}\big|_{1+((i-1)\bmod k)\,=\,l},
\end{equation}
which vanishes for $\ell\ge k$ since groups that are far apart share no
samples. By the law of large numbers applied to the cyclostationary
sequence:
\begin{equation}
  n\cdot\Var\{\bar W_n\}
  \xrightarrow{n\to\infty}
  \bar R_0 + 2\cdot\sum_{\ell=1}^{k-1}\bar R_\ell,
\end{equation}
since each of the $k$ weights appears exactly once per period.
Substituting into~\eqref{eq:mse_master}:
\begin{equation}
  V_k^{\rm vsw}(m_1,\ldots,m_k) = \frac{1}{k^2\bar\mu_W^2}\left[
    \frac{1}{k}\sum_{l=1}^{k}\Var\{W_l\}
    + 2\cdot\sum_{\ell=1}^{k-1}\bar R_\ell
    \right],
\end{equation}
which is~\eqref{eq:mse_periodic}.
The non-overlapping formula ($\bar R_\ell=0$, groups share no samples)
and the sliding-window formula ($\bar R_\ell=R_\ell$, single weight $m_l=m$
for all $l$) follow immediately as special cases. \qed

\section*{Appendix D -- The Density of States}
\renewcommand{\theequation}{D.\arabic{equation}}
    \setcounter{equation}{0}

\subsection*{D.1 General formula}
\label{app:omega_general}

Let $U:\calX\to[0,\infty)$ be a measurable energy function on a state space
$A\subseteq\reals^d$.  Recall that the density of states
$\Omega_k(u)$ is defined by
\begin{equation}\label{eq:Omega_def_app}
  \Omega_k(u) \;=\; \frac{d}{du}\,\mathrm{vol}\bigl\{\bx\in \calX^k:
  U_k(\bx)\le u\bigr\},
\end{equation}
the derivative of the sub-level-set volume with respect to the energy level $u$.
It is a purely geometric quantity, independent of $\beta$.

\paragraph{The $k=1$ case: co-area formula.}
For $k=1$ and $d=1$ (one-dimensional state space), the level set
$\{x\in \calX: U(x)=u\}$ is a finite collection of points for generic $u$.
At each such point $x_0$, the contribution to $\Omega_1$ by the
co-area formula (or simple change of variables) is $1/|U'(x_0)|$:
\begin{equation}\label{eq:Omega1_coarea}
  \Omega_1(u) \;=\; \sum_{x_0:\,U(x_0)=u} \frac{1}{|U'(x_0)|}.
\end{equation}
For $d>1$, the sum is replaced by an integral over the level set
with $(d-1)$-dimensional surface measure $d\sigma$:
\begin{equation}
  \Omega_1(u)=\int_{\{U(x)=u\}}|\nabla U(x)|^{-1}\,\mbox{d}\sigma(x).
\end{equation}
When $A$ is discrete (see Remark~\ref{rem:discrete}), $\Omega_1(u)$
is simply the number of states with energy $u$, and $\Omega_k(u)$
is its $k$-fold discrete self-convolution.

\paragraph{The $k\ge 2$ case: convolution.}
Since $U_k(\bx)=\sum_{i=1}^k U(x_i)$ is a sum of $k$ independent
copies (under the flat measure on $\calX^k$), the density of states
satisfies the convolution identity
\begin{equation}\label{eq:Omega_conv}
  \Omega_k(u) = (\underbrace{\Omega_1 * \cdots * \Omega_1}_{k})(u)
  = \int_0^u \Omega_{k-1}(t)\Omega_1(u-t)\,\mbox{d}t,
\end{equation}
the $(k-1)$-fold convolution of $\Omega_1$ with itself on $(0,\infty)$.
This follows directly from the representation
$\Omega_k(u) = \int_{\calX^k}\delta(U_k(\bx)-u)\,\mbox{d}\bx$
and the independence of the summands.

The relation~\eqref{eq:Omega_conv} reduces the computation of
$\Omega_k$ to a one-dimensional convolution for any $k$, regardless
of the dimension $d$ of the state space.  Concretely, computing
$\Omega_k$ requires at most $k-1$ successive one-dimensional
integrations --- far cheaper than the $k$-dimensional integration
over $\calX^k$ that a general weight function $m(\bx)$ not of group-energy
form would require.  Moreover, the $k$-fold convolution can be
implemented exactly via just \emph{two} Fourier integrations,
regardless of $k$: compute the Fourier transform
\begin{equation}\label{eq:FT_Omega}
  \hat\Omega_1(\omega) = \int_0^\infty e^{i\omega u}\,\Omega_1(u)\,\mbox{d}u,
\end{equation}
raise pointwise to the $k$-th power to obtain
$[\hat\Omega_1(\omega)]^k = \widehat{\Omega_k}(\omega)$
(by the convolution theorem), and recover $\Omega_k$ by the inverse
Fourier transform
\begin{equation}\label{eq:IFT_Omega}
  \Omega_k(u) = \frac{1}{2\pi}\int_{-\infty}^\infty
  e^{-i\omega u}[\hat\Omega_1(\omega)]^k\,\mbox{d}\omega.
\end{equation}
The cost is thus two integrations plus one pointwise power, independent
of $k$.  When $\Omega_1(u)$ is proportional to a standard density
whose convolutions are in closed form (e.g., the Gamma family),
$\Omega_k$ is explicit and no numerical integration is needed.

As a check, one always has
\begin{equation}\label{eq:Omega_check}
  \int_0^\infty \Omega_k(u)\,e^{-\beta u}\,\mbox{d}u = Z^k(\beta),
\end{equation}
since the left side is $\int_{\calX^k} e^{-\beta U_k(\bx)}\,\mbox{d}\bx
= [Z(\beta)]^k$.

\subsection*{D.2 Saddlepoint approximation for large $k$}
\label{app:saddlepoint}

For large $k$, the exact convolution formula~\eqref{eq:Omega_conv}
can be approximated in closed form via the \emph{saddlepoint method},
which is applicable to any energy function satisfying mild regularity
conditions.

Define the log-Laplace transform of $\Omega_1$ as
\begin{equation}\label{eq:CGF}
  K(\theta) \;\dfn\; \ln\!\int_0^\infty e^{\theta u}\,\Omega_1(u)\,\mbox{d}u,
\end{equation}
defined for $\theta$ in a neighborhood of $0$.
Under the tilted measure on $(0,\infty)$ with density proportional
to $\Omega_1(u)e^{\theta u}$, the mean and variance of $u$ are
$K'(\theta)$ and $K''(\theta)$ respectively.

For a given $u>0$, the \emph{saddlepoint} $\hat\theta=\hat\theta(u)$ is
the solution to
\begin{equation}\label{eq:saddlepoint_eq}
  k\,K'(\hat\theta) = u,
\end{equation}
i.e., the tilt that makes the mean of $k$ copies equal $u$.
The saddlepoint approximation to $\Omega_k(u)$ is then
\begin{equation}\label{eq:saddlepoint}
  \Omega_k(u) \;\approx\;
  \sqrt{\frac{k}{2\pi K''(\hat\theta)}}\;
  \exp\!\bigl[k\,K(\hat\theta) - \hat\theta\,u\bigr].
\end{equation}
This approximation is accurate to relative error $O(1/k)$ uniformly
in $u$, far better than the $O(1/\sqrt{k})$ accuracy of the Gaussian
(CLT) approximation.  In particular, it captures the correct
exponential tail behavior of $\Omega_k(u)$ for $u$ far from its mode
$k K'(0)$, where the CLT fails.

\subsection*{D.3 Special cases}
\label{app:omega_examples}

We derive $\Omega_k(u)$ for the three energy functions used in
Section~\ref{sec:numerics}, as special cases of the general
formulae~\eqref{eq:Omega1_coarea} and~\eqref{eq:Omega_conv}.

\paragraph{Example 1: $U(x)=|x|$, $\calX=\reals$.}
The level set $\{|x|=u\}=\{-u,u\}$ has two points, each with
$|U'(x_0)|=1$, so $\Omega_1(u)=2$.
Since $\Omega_1$ is constant, $\Omega_k$ is the $k$-fold convolution
of the constant function $2$ with itself on $(0,\infty)$.
By the general power-law formula~\eqref{eq:Omega_k_powerlaw}
with $\gamma=1$ (or by the self-contained simplex derivation
in Section~A.4 below):
\begin{equation}\label{eq:Omega_k_exp}
  \Omega_k(u) = \frac{2^k\,u^{k-1}}{(k-1)!},\quad u>0.
\end{equation}
Check: $\int_0^\infty \Omega_k(u)\,e^{-u}\,\mbox{d}u = 2^k = [Z(\beta{=}1)]^k$.\quad
(A self-contained derivation of~\eqref{eq:Omega_k_exp} via the simplex
volume argument is given at the end of this appendix.)

\paragraph{Example 2: $U(x)=|x|^{3/2}$, $\calX=\reals$.}
The change of variables $u=|x|^{3/2}$ gives $|x|=u^{2/3}$ and
$|dx|=\tfrac{2}{3}u^{-1/3}\,\mbox{d}u$, so by~\eqref{eq:Omega1_coarea}:
\begin{equation}
  \Omega_1(u) = 2\cdot\frac{1}{|U'(x_0)|}\bigg|_{x_0=u^{2/3}}
  = 2\cdot\frac{1}{\tfrac{3}{2}u^{1/3}} = \tfrac{4}{3}\,u^{-1/3}
  = \tfrac{4}{3}\,u^{2/3-1},\quad u>0.
\end{equation}
Since $\Omega_1(u)=C_1\,u^{a-1}$ with $C_1=\tfrac{4}{3}$ and $a=2/3$,
the convolution~\eqref{eq:Omega_conv} yields $\Omega_k$ proportional
to the $\mathrm{Gamma}(ka,1)$ density:
\begin{equation}\label{eq:Omega_k_stretch}
  \Omega_k(u) = \frac{[C_1\,\Gamma(a)]^k}{\Gamma(ka)}\,u^{ka-1}
  = \frac{\bigl[\tfrac{4}{3}\Gamma\!\left(\tfrac{2}{3}\right)\bigr]^k}
         {\Gamma(2k/3)}\,u^{2k/3-1},\quad u>0.
\end{equation}
Check: $\int_0^\infty \Omega_k(u)e^{-u}\,\mbox{d}u
= [\tfrac{4}{3}\Gamma(\tfrac{2}{3})]^k = [Z(\beta{=}1)]^k$.

\paragraph{Example 3: $U(x)=(x^2-1)^2$, $\calX=\reals$.}
The level set $\{(x^2-1)^2=u\}$ has up to four points: setting
$x^2=1\pm\sqrt{u}$ gives $x=\pm\sqrt{1+\sqrt{u}}$ (always real) and
$x=\pm\sqrt{1-\sqrt{u}}$ (real only for $0<u<1$).  With
$U'(x)=4x(x^2-1)$, formula~\eqref{eq:Omega1_coarea} gives
\begin{equation}
  \Omega_1(u) = \sum_{x_0:\,U(x_0)=u}\frac{1}{|4x_0(x_0^2-1)|},
\end{equation}
which has no closed form.  For $k\ge 2$, $\Omega_k$ is computed by
$(k-1)$ successive applications of~\eqref{eq:Omega_conv} using
numerical integration or fast Fourier transform (FFT) convolution, each a one-dimensional
operation.  One verifies $\int_0^\infty\Omega_1(u)e^{-u}\,\mbox{d}u
= Z(\beta{=}1)\approx 1.974$.

For large $k$, the saddlepoint approximation~\eqref{eq:saddlepoint}
applies with $K(\theta)=\ln\int_0^\infty e^{\theta u}\Omega_1(u)\,\mbox{d}u$,
which is evaluated numerically.

\paragraph{General power-law energy $U(x)=|x|^\gamma$, $\calX=\reals$.}
For any $\gamma>0$, the same argument as Example~2 gives
$\Omega_1(u)=(2/\gamma)\,u^{1/\gamma-1}$, and the convolution formula
yields
\begin{equation}\label{eq:Omega_k_powerlaw}
  \Omega_k(u) = \frac{[Z(\beta{=}1)]^k}{\Gamma(k/\gamma)}\,u^{k/\gamma-1},
  \quad u>0,
\end{equation}
where $Z(\beta{=}1)=(2/\gamma)\Gamma(1/\gamma)$.  The group energy
$U_k(\bx)$ follows a $\mathrm{Gamma}(k/\gamma, 1)$ distribution under
$p^{\otimes k}$.  Examples~1 and~2 correspond to $\gamma=1$ and
$\gamma=3/2$ respectively.

\subsection*{D.4\quad Simplex volume derivation for Example~1}

For completeness we give a self-contained derivation of~\eqref{eq:Omega_k_exp}
not relying on the convolution formula.

Set $V_i=|x_i|$ for each $i$ and sum over all $2^k$ sign
combinations of $(x_1,\ldots,x_k)$:
\begin{equation}
  \Omega_k(u) = 2^k\,\frac{d}{du}\,\mathrm{vol}\bigl\{(V_1,\ldots,V_k)
  \in[0,\infty)^k:\textstyle\sum_{i=1}^k V_i\le u\bigr\}.
\end{equation}
The change of variables $W_i=\sum_{j=1}^i V_j$ maps
$\{V_i\ge 0,\,\sum_i V_i\le u\}$ bijectively to
$\{0\le W_1\le\cdots\le W_k\le u\}$ with Jacobian $1$.
The latter is one of the $k!$ congruent chambers tiling $[0,u]^k$,
so its volume is $u^k/k!$.  Differentiating:
$\mathrm{vol}_{k-1}(\{\sum_{i=1}^k V_i=u\})=u^{k-1}/(k-1)!$, giving
$\Omega_k(u)=2^k u^{k-1}/(k-1)!$.\qed

\end{document}